\documentclass[12pt]{article}

\usepackage{graphicx}
\usepackage{color}
\usepackage{xcolor}

\usepackage{amsmath,amssymb,epsf,epic,amsthm}

\usepackage{dsfont}

\newcommand{\un}{{\mathbb I}}
\newcommand{\ra}{\rightarrow}
\newcommand{\tr }{\mbox{tr}}

\newcommand{\bra}{\langle} 

\newcommand{\ket}{\rangle}

\renewcommand{\i}{{\rm i}}

\newcommand{\D}{{\mathbb D}}
\newcommand{\Dd}{{\mathcal D}}
\newcommand{\M}{{\mathbb M}}
\newcommand{\B}{{\mathbb B}}
\newcommand{\Bb}{{\mathcal B}}
\newcommand{\Aa}{{\mathcal A}}
\newcommand{\Mm}{{\mathcal M}}
\newcommand{\K}{{\mathcal K}}
\renewcommand{\vec}[1]{\mathbf{#1}}

\newcommand{\be}{\begin{equation}}
\newcommand{\ee}{\end{equation}}
\newcommand{\bea}{\begin{eqnarray}}
\newcommand{\eea}{\end{eqnarray}}

\newcommand{\ffi}{\varphi}

\newcommand{\ep}{\hfill  {\vrule height 10pt width 8pt depth 0pt}}

\newcommand{\grintl}{[\kern-.18em [}
\newcommand{\grintr}{]\kern-.18em ]}

\newcommand{\id}{\mathds{1}}

\newcounter{resultcounter}[section]

\newtheorem{thm}[resultcounter]{Theorem}
\newtheorem{lem}[resultcounter]{Lemma}
\newtheorem{prop}[resultcounter]{Proposition}
\newtheorem{cor}[resultcounter]{Corollary}

\newtheorem{rem}[resultcounter]{Remark}
\newtheorem{rems}[resultcounter]{Remarks}

\def\cA{{\cal A}}  \def\cC{{\cal C}}
  \def\cF{{\cal F}}
 \def\cH{{\cal H}} 
 \def\cK{{\cal K}} 
\def\cM{{\cal M}}  \def\cO{{\cal O}}
\def\cP{{\cal P}}

\newcommand{\R}{{\mathbb R}}

\newcommand{\N}{{\mathbb N}}

\newcommand{\C}{{\mathbb C}}
\newcommand{\Z}{{\mathbb Z}}

\renewcommand{\P}{{\mathcal P}}
\newcommand{\I}{{\mathbb I}}

\newcommand{\Ss}{{\mathbb S}}
\newcommand{\F}{{\mathcal F}}

\newcommand{\Ee}{{\mathcal E}}

 \def\idtyty{{\mathchoice {\mathrm{1\mskip-4mu l}} {\mathrm{1\mskip-4mu l}} %
{\mathrm{1\mskip-4.5mu l}} {\mathrm{1\mskip-5mu l}}}}

\begin{document}
\title{Thermalization of Fermionic  Quantum Walkers}

\author{ Eman Hamza\footnote{Faculty of Science, Cairo University, Cairo 12613, Egypt} \and Alain Joye\footnote{ Universit\'e Grenoble Alpes, CNRS, Institut Fourier, F-38000 Grenoble, France} }

\date{ }

\maketitle
\vspace{-1cm}

\thispagestyle{empty}
\setcounter{page}{1}
\setcounter{section}{1}

\setcounter{section}{0}

\abstract{ We consider the discrete time dynamics of an ensemble of fermionic quantum walkers moving on a finite discrete sample, interacting with a reservoir of infinitely many quantum particles on the one dimensional lattice. The reservoir is given by a fermionic quasifree state, with free discrete dynamics given by the shift, whereas the free dynamics of the non-interacting quantum walkers in the sample is defined by means of a unitary matrix. The reservoir and the sample exchange particles at specific sites by a unitary coupling and we study the discrete dynamics of the coupled system defined by the iteration of the free discrete dynamics acting on the unitary coupling, in a variety of situations. In particular, in absence of correlation within the particles of the reservoir and  under natural assumptions on the sample's dynamics, we prove that the one- and two-body reduced density matrices of the sample admit large times limits characterized by the state of the reservoir which are independent of the free dynamics of the quantum walkers and of the coupling strength. Moreover, the corresponding asymptotic density profile in the sample is flat and the correlations of number operators have no structure, a manifestation of thermalization.

}

\thispagestyle{empty}
\setcounter{page}{1}
\setcounter{section}{1}

\setcounter{section}{0}

\section{Introduction}
Quantum walks, in their various guises, deterministic or random, are at the crossroad of quantum physics, quantum computing, non-commutative probabilities and analysis, see e.g. the reviews \cite{Ke, Ko, V-A, J4, ABJ2}. A quantum walk is essentially a unitary operator on a Hilbert space with basis elements associated to the vertices of an underlying graph, whose matrix elements couple nearest neighbours of the graph only. This operator can be viewed as the one time step unitary discrete time evolution of a quantum particle with spin hopping on the sites of the underlying graph, its configuration space. One gets a discrete time quantum dynamical system by iteration of this unitary operator.

Quantum walks have been the object of many works in the recent years, from several perspectives. To give a few examples, some papers explore their ability to provide models for the dynamics of actual quantum systems, and others describe their role in the elaboration of quantum algorithms. Some works  study the relations between quantum walks and classical random walks, the formers being considered as the quantum counterparts of the latters \cite{GVWW}, while others analyze the spectral and transport properties they possess as discrete quantum dynamical systems, or their links with CMV matrices related to orthogonal polynomials on the unit circle \cite{HJS, JM, ASWe, J3, ABJ, ABJ2}.  From the point of view of quantum mechanics, all these works consider a single, sometimes fictitious, quantum particle with spin, or quantum walker, on the configuration space. See however \cite{A et al.} for an analysis of two coupled quantum walkers. 

 By contrast, our aim is to study the collective large times dynamical behaviour of an ensemble of quantum walkers on a graph in the framework of many body quantum statistical physics, starting with the basic thermalization properties of this ensemble when put in contact with an infinite reservoir of quantum particles. The dynamics of the  full model consisting in the collection of walkers coupled to the infinite reservoir is discrete in time, and thus characterized by a one time step unitary operator  on the relevant Fock space, in keeping with that of single quantum walkers.   One motivation for this problem stems from the fact that some of the quantum systems approximated successfully by a one body quantum walk are genuinely fermionic many body quantum systems, \cite{CC}. Another reason to investigate the statistical mechanics of quantum walks comes from the role models of interacting classical random walks, or exclusion processes, play in non equilibrium statistical mechanics. Hence, with quantum walks, we are heading towards a quantum version of such models, where interactions between quantum walkers are replaced by the Pauli exclusion induced by the choice of fermionic statistics. 
 
More precisely, the model we consider has the following features. The configuration space of the quantum walkers, called sample, is given by a finite one dimensional lattice. The particles in the sample are not driven by a Hamiltonian, therefore we work within the grand canonical formalism. The quantum walkers are thus considered as noninteracting fermionic particles characterized by the one particle one time step unitary dynamics, which defines their free dynamics. The infinite reservoir of particles consists of fermionic noninteracting quantum walkers as well on the infinite one dimensional discrete lattice. The one particle free dynamics in the reservoir is simply given by the shift. The reservoir is initially in a quasi free state characterized by a positive density operator $\Sigma$, ${\mathbb O} < \Sigma\leq \I$, defined on the one particle Hilbert space which describes the correlations within the reservoir. The interaction we choose between these two species of fermions allows for transformations of particles of one kind into the other, so that the number of particles within the sample may vary with time. Moreover, the coupling takes place at specific sites of the sample and of the reservoir. The unitary coupling considered is the exponential of $i$ times a creation operator in the sample times an annihilation operator in the reservoir plus hermitian conjugate, similar to the dipole interaction between particles in the rotating wave approximation. The one time step dynamics of the coupled system living on the tensor product of their respective fermionic Fock spaces is then defined by the composition of the unitary coupling just described, followed by the one time step decoupled free dynamics in the sample and in the reservoir. Again, iteration provides us with a discrete time quantum dynamical system on the Fock space of the coupled system. 
The thermalization process we are interested is encoded in the large time behaviour of the reduced density matrix of the fermionic quantum walkers in the sample, which is the main focus of this work.

A few remarks are in order:  while the fermionic nature of the quantum walkers is motivated by the goals stated above, the choice of fermionic reservoir is largely dictated by the fact that it makes the mathematics simpler; moreover, for the thermalization process we are interested in, the reservoir statistics should not matter much for the large time properties of the sample.  Similarly, the coupling between the two fermion species is admittedly hardly physical, and essentially motivated by the fact that it provides a simple mechanism of exchange of particles between the sample and the reservoir, suitable for the grand canonical formalism we adopt.  Also, we emphasize that since the system is not described by a Hamiltonian, there is no {\it a priori} notion of thermal Gibbs state. The key observables of the theory are, instead, the number of particles at the various sites of the sample. 

Let us  informally describe our main results. After  setting the stage in the rest of the present section, we start the  analysis  in Section \ref{secshift} by considering  a simple exactly solvable situation  in which the free dynamics of the quantum walkers is given by a shift on the discrete circle, with arbitrary density operator $\Sigma$ in the reservoir. 
This special case allows for a detailed treatment which sets some milestones to compare to when more general situations are considered later on.
We compute exactly the one-body and two-body reduced density matrices in the sample as a function of time in Theorems \ref{thm:1DM} and \ref{thm:2DM} and we deduce that the asymptotic particle density profile is flat in the sample, with a value given by the particle density the reservoir. Moreover, the spatial correlations in particle numbers depend on the distance between the particles only, see Corollary \ref{correlations}.  Furthermore, we show in Theorem \ref{qfeven} that the infinite time limit of the whole reduced density matrix on the sample exists and is quasifree, with density parametrized by $\Sigma$. When $\Sigma=\sigma \idtyty$, the asymptotic state  in the sample turns out to be a Gibbs state in the total number operator, depending on the particle density $\sigma$ only, and neither on the strength of the coupling, nor on the size of the sample. 

This remark makes the transition to Section \ref{ris} devoted to the case $\Sigma=\sigma \idtyty$, which corresponds to the absence of correlations within the reservoir. This case corresponds to a repeated interaction dynamics for the particles in the sample and we  obtain in particular the time dependence of all $p-$body reduced density matrices in the sample, Proposition \ref{tdpbdenmat}. Moreover, when the sample contains initially no particle, we show that the distribution of the number of particles in the sample is a binomial law, whose time-dependent characteristics we specify, see Corollary \ref{corrn=p}. Finally, the dynamics of the particle flux observable in and out of the reservoir and  
 its asymptotic saturation properties are described in Proposition \ref{fluxsat}.

So far, the dynamics in the sample is extremely regular, since it is given by the shift. In Section \ref{perdynsam}, we eventually turn to quantum walkers in the sample whose free dynamics is arbitrary, in contact with a reservoir characterized by a constant density and no correlations, $\Sigma=\sigma \idtyty$. Building up on the previous sections, we determine the explicit  time dependence of the one-body and two-body density matrices in the sample in Theorems \ref{thm:1DMpertub} and  \ref{perdyn2b}. Moreover we prove that if the dynamics in the sample is mixing enough, a property expressed as a spectral hypothesis, the long time limits of the one- and two-body reduced density matrices exist and coincide with those obtained for the shift in the sample. They are given by $\sigma$, respectively $\sigma^2$, times the identity, and are thus completely  independent of the sample dynamics and of the strength of the coupling, see Corollaries \ref{44} and \ref{47}. This final section ends with an application to coined quantum walks, where it is shown that the spectral assumptions alluded to above hold true generically.  

As a consequence, this result shows that, generically, fermionic quantum walkers in contact with a reservoir thermalize to the same asymptotic state which only depends on the density of particles in the reservoir. This is true in particular for random quantum walks of the kind considered in \cite{JM}, which are known to display Anderson localization and to give rise to finite volume exponential decay estimates of the evolution kernel. Even though these features are likely to provide some structure in the density profile and correlations, the thermalization process considered washes out any spatial structure, in keeping with similar properties in the Hamiltonian framework, see \cite{FS}.

{\bf Acknowledgments :}\\
This work has been partially supported by the LabEx PERSYVAL-Lab (ANR-11-LABX- 0025-01) funded by the French program Investissement d'avenir, and a joint Science and Technology Development Fund \& Institut Fran\c cais d'Egypte grant (STDF-IFE 2015).
%
\subsection{Notation}
We  start by fixing the notation used throughout this paper. Our Hilbert spaces will be complex and separable, and the scalar product $\bra \cdot\, | \, \cdot \ket$ is linear is the rightmost variable. We will denote by $\cH^{\otimes p}$, the p-fold tenser product of a Hilbert space $\cH$ . We denote the antisymmetric tensor product of a set of $q\leq p$ vectors, $u_1,..,u_q\in\cH$, by
\be\label{wedge}
u_1\wedge  ..\wedge u_q=\frac{1}{\sqrt{q!}}\sum_{\pi\in S_q} \epsilon_{\pi}  u_{\pi(1)}\otimes ..\otimes u_{\pi(q)},
\ee
where $S_q$ is the group of permutations of $\{1,2,..,q\}$ and $\epsilon_{\pi}$ is the signature of the permutation $\pi$. Such vectors are called \{elementary vectors\}.
The p-fold antisymmetric tensor product of $\cH$ denoted by $\cH^{\wedge p}$, is defined as the closure of the subspace of $\cH^{\otimes p}$ generated by $u_1\wedge  ..\wedge u_p$, where $u_1,..,u_p\in\cH$ form an orthonormal set of linearly orthonormal vectors, in which case $u_1\wedge  ..\wedge u_p$ is of norm one.

The projections $\P_A^{(p)}$ onto $\cH^{\wedge p}$ is given by
\be\label{defproj}
\P_A^{(p)}=\frac{1}{p!}\sum_{\pi\in S_p} \epsilon_{\pi} \Theta(\pi),
\ee
where $\Theta$, the natural representation of the permutation group $S_p$, is given by \cite{Derezinski, DFP}

\be
\Theta(\pi) u_1\otimes\dots\otimes u_p=u_{\pi(1)}\otimes \cdots\otimes u_{\pi(p)}.
\ee
The antisymmetric Fock space is defined as
\be
\F_-(\cH)=\C\oplus \bigoplus_{p=1}^{\mbox{\scriptsize dim} \cH} \cH^{\wedge p},
\ee
where dim$\cH$ may be infinite, and the vacuum vector is denoted by $|\Omega\ket$.
Finally, recall that fermionic creation operators $c^*$ are defined by their action on any elementary vector $u_1\wedge  ..\wedge u_q$ of $\cH^{\wedge q}$, and by linearity on $\F_-(\cH)$, in the following way. For any $\ffi\in \cH$, $c^*(\ffi)$ acts as
\be
c^*(\ffi)u_1\wedge  \cdots \wedge u_q=\ffi\wedge u_1\wedge  \cdots \wedge u_q, 
\ee
so that
\be
u_1\wedge \cdots  \wedge u_q=c^*(u_1)\cdots c^*(u_q)|\Omega\ket,
\ee
 and $c(\ffi)$ is the adjoint of $c^*(\ffi)$, such that $c(\ffi)|\Omega\ket=0$. More generally
\be
c(\ffi)u_1\wedge  \cdots \wedge u_q=\sum_{j=1}^q(-1)^{j-1}\bra\ffi | u_j\ket u_1\wedge  \cdots \wedge u_{j-1}\wedge u_{j+1}\wedge \cdots \wedge u_q.
\ee
These operators satisfy the CAR relations
\be
\{c(\ffi),c(\chi)\}=\{c^*(\ffi),c^*(\chi)\}=0, \ \ \ \{c(\ffi),c^*(\chi)\}=\bra \ffi | \chi\ket \I.
\ee

 In our model,  the reservoir Hilbert space is $\cH_r=\ell^2(\Z)$ where $\Psi_j$, $j\in\Z$, are the canonical basis vectors. The Hilbert space of the small system, or sample, is $\cH_s=\ell^2(\{0,1,2,..,d-1\})\simeq \C^{d}$,  with $e_j$, $j\in \{0,1,2,..,d-1\}$ being its canonical basis vectors. We denote by $a^*$, $a$ the fermionic creation and annihilation operator on $\F_-(\cH_s)$ and by $\Aa$ the $C^*$ algebra generated by $\{ a^*(\psi), a(\psi): \psi\in\cH_s\}$.    On the other hand, $b^*$ and $b$ denote the fermionic creation and annihilation operators on $\F_-(\cH_r)$, while $\Bb$ is the $C^*$ algebra generated by $\{ b^*(\phi), b(\phi): \phi\in\cH_r\}$.   Hence, the Fock space of the composite system is given by the tensor product of the (anti-symmetric) Fock spaces of the two sub-systems, i.e.
\be
\F= \F_-(\cH_r)\otimes \F_-(\cH_s).
\ee
We write $a^{\#}_j$ for $a^{\#}(e_j)$ and $b^{\#}_j=b^{\#}(\Psi_j)$, with $\#=*$ or nothing. 
The number operators on $\F_-(\cH_r)$ and $\F_-(\cH_s)$ are respectively given by 
\bea
N_r=\sum_i b_i^*b_i=\sum_{i\in\Z} n_i^r, \ \ \mbox{and } \  
N_s=\sum_i a_i^*a_i=\sum_{i=0}^{d-1} n^s_i\label{numbersam}.
\eea
The operator $N_r$ is unbounded on $\F_-(\cH_r)$ with maximal domain
\be
\Dd(N_r)=\big\{ \Phi=(\phi^0,\phi^1,...)\in\F_-(\cH_r): \sum_{p\geq 0}p^2\|\phi^p\|_{\cH^{\wedge p}}^2<\infty\big\}, \ \ \mbox{where} \ \ \phi^p\in \cH^{\wedge p}.
\ee

\subsection{The Dynamics}

The discrete dynamics of the system is characterized by the one time step unitary operator $U$ on $\F=\F_-(\cH_r)\otimes \F_-(\cH_s)$ given by
\be\label{coupleddyn}
U= U_F K,
\ee
where the {\it free} one time step dynamic $U_F$ on $\F$  is given by the tensor product of $U_s$ the {free} unitary dynamic on $\F_-(\cH_s)$ and $U_r$ the  {free} unitary dynamic on $\F_-(\cH_r)$, i.e.  
\be\label{freedyn}
U_F= U_r\otimes U_s.
\ee
The one time step free evolution $\tau_r$ on  $\Bb$  is  defined as 
\bea
\tau_r(A)=U_r^* A U_r && \text{ for } A\in \Bb,
\eea
and the one time step free evolution $\tau_s$ is defined similarly on $\Aa$ using the dynamics $U_s$.  The one time step free evolution is naturally defined by the tensor product $\tau_r\otimes\tau_s$ on the tensor product of $C^*$ algebras $\Bb\otimes \Aa$; recall that dim$\cH_s<\infty$.

Further introducing a coupling between the two systems given by a unitary operator $K$ acting on $\F$, we define the time evolution  $\tau^t$ on  $\Bb\otimes \Aa$ at $t\in\Z$, as
\bea
\tau^t(A_r\otimes A_s)=U^{*t} (A_r\otimes A_s) U^t && \text{ for } A_r\in \Bb \text{  and  } A_s\in \Aa,
\eea
where, for one time step,
\be
U^{*} (A_r\otimes A_s) U =K^{*} (U_r^* A_r U_r\otimes U_s^* A_s U_s) K.
\ee

 In this paper, the free dynamics $U_r$  is defined as the second quantization of the shift $S$ on $\cH_r$, $U_r=\Gamma(S)$,  with $S$ is given by 
\be\label{eq:dynreserve} 
S \Psi_{j}=\Psi_{j-1}
\ee
where $\Psi_j$ are the canonical basis vectors of $\ell^2(\Z)$.  
On the other hand, we consider the free dynamics in the sample to be the second quantization of an arbitrary unitary dynamics $W$ on $\cH_s$. 
 More precisely, we consider the dynamics $U_s$  to be given by  
\be\label{def:perturbeddyn}
U_s= \Gamma (W), \ \ \mbox{where }\ \ W:\cH_s\ra\cH_s \ \ \mbox{is unitary}.
\ee
For convenience, we consider $\cH_s$ to be supplemented by periodic boundary conditions so that $e_d\equiv e_0$ and $a^\#_{r+md}\equiv a^\#_{r}$, for all $m\in\N$ and $r\in\{0,1,\dots, d-1\}$.
The coupling between the two systems is given by $K=K_{0}$, where the unitary operators $K_j$, $j\in\N$ acting on $\F$ are 
defined by
\be\label{coupling}
K_{j}=K_{j}(\alpha)= \exp[-i\alpha(b_j^*\otimes a_j+b_j\otimes a_j^*)],
\ee
where $\alpha\in \R$ plays the role of a coupling constant. The action of $K$ consists in transforming fermions from the reservoir to fermions from the sample, and vice versa, when they both sit on the site labeled by zero in their respective Hilbert spaces.  

\subsection{The initial state}

The initial state of the reservoir $\omega_\Sigma$  is a gauge-invariant quasi-free  state satisfying all $n ,m\in \N$, and 
all $\phi_1, .., \phi_m, \psi_1,.., \psi_n \in \cH_r$,
\be\label{defquasifree}
\omega_\Sigma(b^*(\phi_m)....b^*(\phi_1)b(\psi_1)..b(\psi_n))=\delta_{nm} \det\{(\psi_j, \Sigma\phi_k)\}.
\ee
with a self-adjoint  density $\Sigma$, ${\mathbb O} \leq \Sigma\leq \I$ on $\cH_r$. Our choice of $\Sigma$ is motivated by requiring that the state $\omega_\Sigma$ has the following properties:\\
 
[{\it i}\,] The state $\omega_\Sigma$ is invariant under the free time evolution of the reservoir $\tau_r$, where
\be
\omega_\Sigma \circ \tau_r=\omega_{S \Sigma S^*}\equiv \omega_\Sigma.
\ee
This implies that $[\Sigma,S]=0$, which means that the matrix elements of $\Sigma$ can be written as  
\be
\Sigma_{jk}=\sigma(k-j),
\ee
 for a function $\sigma:\Z\to \C$. In order for  $\Sigma$ to be self adjoint, we require that for all $k\in\Z$,
 \be
\sigma(k)=\overline{\sigma(-k)}
\ee
 and since $0\leq \Sigma\leq 1$, we have that $\| \sigma\|_\infty \leq 1$ and $\sigma\in l^2(\Z)$.   
  \vspace{0.1cm}

[{\it ii}\,] For the state $\omega_{\Sigma}$ to have a finite density of particles, we need for all $k\in\Z$
$$\omega_\Sigma(n^r_k)=\omega_\Sigma(b_k^*b_k)=\Sigma_{kk}=\sigma(0)>0.$$

On the other hand, the initial state of the small system is characterized by a density matrix $\rho$, i.e. a positive trace one operator $\rho:\F_-(\cH_s)\to \F_-(\cH_s)$ such that the expectation of any observable $A$ on $ \F_-(\cH_s)$ is given by
\be
\rho(A)=\tr_{\F_-(\cH_s)} (\rho A).
\ee  
Along with $\rho$, we will consider also the $p$-body reduced density matrix , $\rho^{(p)}$ on $\cH_s^{\wedge p}$, the matrix elements of which are given by
\bea\label{denistyreduced}
\bra e_{j_1}\wedge..\wedge e_{j_p}|\rho^{(p)} e_{k_1}\wedge ..\wedge e_{k_p}\ket
=\tr_{\F_-(H_s)} (\rho a_{k_1}^*..a_{k_p}^* a_{j_p}..a_{j_1}).
\eea
More precisely, we will be mainly interested in the behaviour in time of the (reduced) density matrix on the sample, defined for all $t\in \N$ by
\bea
\rho_t(A)&=&(\omega_\Sigma \otimes \rho) \circ \tau^t (\idtyty\otimes A) \ \ \mbox{for all observables $A\in \cA$ and }  \\
 \rho^{(p)}_t(A)&=&(\omega_\Sigma \otimes \rho^{(p)}) \circ \tau^t (\idtyty\otimes A) \ \ \mbox{for all $p-$body observables $A$ on $\cH_s^{\wedge p}.$}  
\eea

\subsection{The Flux}
A natural observable in this context is the flux giving the variation in the  number of particles in the reservoir in one time step, that  is formally defined as
\be
\Phi_r=U^*N_rU-N_r.
\ee 
A simple calculation shows that
\be\label{def:flux}
\Phi_r=\sin^2(\alpha)(\idtyty\otimes n^s_0-n_0^r\otimes \idtyty)+i\sin(\alpha)\cos(\alpha)(b_0^*\otimes a_0-b_0\otimes a_0^*).
\ee
which is a bounded operator on $\F$. Taking this  as a definition of the flux,
we shall consider its dynamics  in certain cases below.

%

%
\subsection{General properties}

We start with a few simple and general properties these operators possess, that will be used frequently in the following. 
First, we recall that
\begin{lem}
The one-body reduced density matrix, $\rho^{(1)}$, of a full density matrix $\rho$ on $\cF_-(\cH)$ satisfies  $0\leq \rho^{(1)}\leq 1$ as an operator on $\cH_s$.
\end{lem}

On the other hand,  clearly, $[U_r, N_r]=[U_s,N_s]=0$, so that $[U_F, N]=0$. Similarly, 
\begin{lem} 

For any $j\in \N$,
\bea
[K_{j}, N]=0, \ \ \mbox{so that}\ \  [U, N]=0. 
\eea
Also for all $l,m\in\N$ s.t. $l-m\not\in d\Z$ 
\be\label{k,k}
[K_l,K_m]= [K_l,K^*_m]=0
\ee

\end{lem}
In other words, the operators $U$, $U_F$ and $K_{j}$ given by \eqref{coupleddyn}, \eqref{freedyn} and \eqref{coupling} conserve the total number of particles. 

Since we are dealing with fermions, the creation/annihilation operators are bounded so that we can actually compute $K_j$ by the power series of the exponential, using the fact that $n^r_j$ and $n^s_j$ are projectors:
\begin{lem}\label{kj}
 For any $\alpha\geq 0$, the operators given by \eqref{coupling} can be written as
\begin{align}
K_{j}&= \idtyty+g_\alpha(b_j^*\otimes a_j+b_j\otimes a_j^*)+f_\alpha (n^r_j\otimes \id-\id\otimes n^s_j)^2\\
&=\idtyty+g_\alpha(b_j^*\otimes a_j+b_j\otimes a_j^*)+f_\alpha (n^r_j\otimes (\id-n^s_j)+(\id-n^r_j)\otimes n^s_j),
\end{align}
where $g_\alpha=i\sin(\alpha)$ and $f_\alpha=\cos(\alpha)-1$.
\end{lem}
Consequently, using the periodicity in the index $j$ in the sample, explicit computations yield
\begin{lem}\label{evolkj}
For all $k,j\in\{0,..,d-1\}$ and $s\in\N$, the conjugation of the creation and annihilation operators under the coupling are given by
\begin{align}
&K_{j+sd}^* (\idtyty\otimes a_j )K_{j+sd}=\cos(\alpha) \idtyty\otimes a_j+g_\alpha b_{j+sd}\otimes (1-2n_j^s),\label{ak*k}
\\ &K_{j+sd}^* (\idtyty \otimes a^\#_j) K^*_{j+sd}=\cos(\alpha)\idtyty \otimes a^\#_j-g_\alpha b^\#_{j+sd}\otimes \idtyty,\label{ak*k*}
\\ &K_{j+sd} (\idtyty \otimes a^\#_j) K_{j+sd}=\cos(\alpha) \idtyty \otimes a^\#_j+g_\alpha b^\#_{j+sd}\otimes \idtyty,\label{akk}
\\&K_{j+sd}^* (\idtyty\otimes n^s_j )K_{j+sd}=\cos^2(\alpha) \idtyty\otimes n^s_j+g_\alpha \cos(\alpha)( b_{j+sd}\otimes a^*_j- b^*_{j+sd}\otimes a_j)\nonumber \\
&\hspace{3.45cm}+\sin^2(\alpha) n^r_{j+sd}\otimes\idtyty\label{nk*k},
\end{align}
while for all $j\neq k$,
\be\label{eq:odda}
K^*_{k+sd} (\idtyty \otimes a^{\#}_j)=(\idtyty \otimes a^{\#}_j) K_{k+sd}.
\ee
\end{lem}
\begin{rem}
Similar statements are true for the conjugation of $b_k$ by $K_{j}$, thanks to the symmetry in $a^\#$ and $b^\#$ of $K_j$. 
\end{rem}

\section{ Shift in the Sample}\label{secshift}

With these preliminary considerations behind us, we are in a position to address the time evolution of observables in the sample, assuming to start with, the free dynamics of the sample is the periodic shift: {\i.e.} the unitary operator $W$, is  given by the shift  $S_p$ on $\ell^2(\{0,1,..,d-1\})$ defined as
\bea\label{eq:shiftsample}
S_p e_j=e_{j-1},& \text{with  periodic boundary condition  } S_p e_0=e_{d-1}.
\eea 
The free dynamics on the sample $U_s$  is thus given by the second quantization of the shift $S_p$
\be\label{eq:dynshiftsample}
U_s=\Gamma(S_p).
\ee
 Using the Bogoliubov transform, it is easy to see that  \cite{Derezinski}
\begin{align}\label{period}
U^*_s a^{\#}_i U_s&= a^{\#}_{i+1}  \ \ \ \mbox{with  $a^{\#}_{d}:=a^{\#}_{0}$}\\
U^*_r b^{\#}_i U_r&=b^{\#}_{i+1}.
\end{align} 
Moreover, the evolution of the coupling operator under the free dynamics defined by \eqref{freedyn}, is given for all $j\in\Z$ by
\be\label{evolcoup}
U_F^* K_j U_F=K_{j+1}.
\ee

In order to simplify the expressions, we mainly consider times that are integer multiples of $d$, the number of sites in the sample, see Remark \ref{remrem} ii), though. This prescription allows us to take advantage of the spatial periodicity (\ref{period}) of the creation/annihilation operators.
In what follows, the limits are understood in norm convergence, and we omit symbols $\otimes \idtyty$ and $\idtyty \otimes$ whenever the meaning is clear.

\begin{thm}\label{thm:exact}
For all $k,j\in\{0,..,d-1\}$, the following is true
\bea\label{evolova*a}
\tau^{md}( a_k^* a_j)=\cos^{2m}(\alpha)  a_k^* a_j+ \sin^2(\alpha) \sum_{r=0}^{m-1}\sum_{s=0}^{m-1}(\cos(\alpha))^{2(m-1)-(r+s)} ( b_k^*(r) b_j(s))\nonumber
\\-g_\alpha \cos^m(\alpha) \sum_{r=0}^{m-1} (\cos(\alpha))^{m-1-r} a_j\otimes b_k^*(r)+ g_\alpha \cos^{m}(\alpha)  \sum_{r=0}^{m-1} (\cos(\alpha))^{m-1-r}  b_j(r)\otimes a^*_k,
\eea
where $b_x^{\#}(y)=b_{x+yd}^{\#}$.
For all $\alpha\not\in\{0,\pi\}$,  we have 
\bea
\lim_{m\to\infty} \tau^{md} (a_k^* a_j) = \lim_{m\to\infty}\sin^2(\alpha) \sum_{r=0}^{m-1}\sum_{s=0}^{m-1}(\cos(\alpha))^{2(m-1)-(r+s)} ( b_k^*(r) b_j(s)).
\eea

Moreover, for all integers $p\geq 2$, and all distinct $\{j_1,\cdots, j_p\}$  and all distinct $\{k_1,\cdots, k_p\}$, if $\alpha\not\in\{0,\pi\}$,
\bea\label{pdminfinite}
&&\lim_{m\to\infty} 
\tau^{md}( a_{k_1}^* a_{k_2}^*\cdots a_{k_p}^* a_{j_p} \cdots a_{j_2} a_{j_1})
\nonumber \\
&&= \lim_{m\to\infty}\sin^{2p}(\alpha) \sum_{r_1, \dots ,r_p, s_1,\dots ,s_p=0}^{m-1}(\cos(\alpha))^{2p(m-1)-(r_1+s_1+\dots+r_p+r_p)}\nonumber
 \\&&  \hspace{7cm} \times 
 b_{k_1}^*(r_1) b_{j_1}(s_1).. b_{k_p}^*(r_p) b_{j_p}(s_p).
 \eea
\end{thm}
\begin{proof}
To show the first statement,  we first note that by (\ref{evolcoup}), we have for all $n\in\N$,
\begin{align}
U^{n}=&(U_F K_0)\cdots(U_F K_0)(U_F K_0)(U_F K_0)=(U_F K_0)\cdots (U_F K_0)U_F^2 (U_F^*K_0U_F) K_0\nonumber\\
=&(U_F K_0)\cdots U_F^3 ({U_F^*}^2 K_0 U_F^2) K_1 K_0=U_F^n K_{n-1}\cdots K_1K_0.
\end{align}
Thus, further making use of (\ref{period}), \eqref{k,k} and (\ref{eq:odda}) for $k\neq j$,
\begin{align}
U^{*md} a_k^* a_jU^{md}= \big(K^*_k K^*_{k+d}\dots&K^*_{k+(m-1)d}  a^*_k K^*_{k+(m-1)d}\dots K^*_{k+d}K^*_{k}\big) \nonumber
\\ & \times
\big(K_j K_{j+d}\dots K_{j+(m-1)d} a_j K_{j+(m-1)d}\dots K_{j+d}K_j\big).
 \end{align}
Successive applications of \eqref{akk}, (\ref{ak*k*}),  along with the fact that for all $k\neq j$
\be
K^\#_{k}(b_j) K^\#_{k}=b_j,
\ee 
give

\begin{align}\label{a*evolvedk*}
K^*_k K^*_{k+d}...K^*_{k+(m-1)d} a^\#_k K^*_{k+(m-1)d}... K^*_{k+d}K^*_{k}&=\cos^m(\alpha) a_k^\#-g_\alpha\sum_{r=0}^{m-1} \cos^{m-1-r}(\alpha) b_k^\#(r),\\
\label{aevolvedk}
K_j K_{j+d}...K_{j+(m-1)d} a_j^\# K_{j+(m-1)d}... K_{j+d}K_{j}&=\cos^m(\alpha)  a_j^\#+g_\alpha\sum_{r=0}^{m-1} \cos^{m-1-r}(\alpha) b_j^\#(r).
\end{align}

The first statement of the theorem for $k\neq j$ then follows readily. For $j=k$ we use
\begin{align}\label{evolvns}
U^{*md} ( n^s_j)U^{md}&= K^*_jK^*_{j+d}...K^*_{j+(m-1)d}( n^s_j) K_{j+(m-1)d}... K_{j+d}K_{j}\\ \nonumber
&=(K^*_j\dots K^*_{j+(m-1)d}{a_j}^*
K^*_{j+(m-1)d}\dots K^*_j)
(K_{j}\dots K_{j+(m-1)d}a_jK_{j+(m-1)d}... K_{j})
 \end{align}
to get the required result.
The second statement of the theorem is a direct consequence of the first. 
Finally, note that for all $p\geq 1$,
\begin{align}\label{yyy}
\tau^{md}( a_{k_1}^* &a_{k_2}^*\dots a_{k_p}^* a_{j_p} ..a_{j_2} a_{j_1})=\\ \nonumber
&\left\{\begin{matrix}
\tau^{md}( a_{k_1}^*a_{k_2}^*) \cdots \tau^{md}( a_{k_{p-1}}^*a_{k_p}^*)\tau^{md} (a_{j_p}a_{j_{p-1}})
\cdots  \tau^{md}(a_{j_2}a_{j_1})\hspace{2.5cm}& p \ $even$ \cr
\tau^{md}( a_{k_1}^*a_{k_2}^*) \cdots \tau^{md}( a_{k_{p-2}}^*a_{k_{p-1}}^*)\tau^{md} (a^*_{k_p}a_{j_{p}})\tau^{md} (a_{j_{p-1}}a_{j_{p-2}})
\cdots  \tau^{md}(a_{j_2}a_{j_1})& p \ $odd$
\end{matrix}\right.
\end{align}
where the operators $\tau^{md}(a^\#_ka^\#_j)$ with distinct indices are compositions of operators of the form (\ref{a*evolvedk*}) and (\ref{aevolvedk}). Taking the limit $m\ra\infty$ in each of them allows to deduce the last statement in a similar way as the second one.
\ep

 \end{proof}
\begin{rems}\label{remrem}
i) There is an explicit, though cumbersome, expression also for \\ $\tau^{md}( a_{k_1}^* a_{k_2}^*..a_{k_p}^* a_{j_p} ..a_{j_2} a_{j_1})$ for all
finite $m$, as the proof shows.\\
ii) It is possible also to compute the evolution of observables at time $md+r$,  for any $0<r<d$, making use of the following, with the convention (\ref{period}), and Lemma \ref{evolkj}:
\be
 \tau^{r} (a_k^* a_j)=K_0^*K_{1}^*\cdots K^*_{r-1}a^*_{k+r} a_{j+r}K_{r-1}\cdots K_1K_0.
\ee
- If $k+r<d$ and $j+r<d$, $\tau^{r} (a_k^* a_j)=a_{k+r}^* a_{j+r}$. \\
- If $k+r\geq d$ and $j+r<d$, $\tau^{r} (a_k^* a_j)=K^*_{m(k)}a_{m(k)}^*K^*_{m(k)} a_{j+r}$, where $m(k)=(k+r)-d<r$.\\
- If $k+r\geq d$ and $j+r\geq d$, $\tau^{r} (a_k^* a_j)=K^*_{m(k)}a_{m(k)}^*K^*_{m(k)} K_{m(j)}a_{m(j)}K_{m(j)}$.\\
iii) The evolution of observables that contain an odd number of operators $a_j^\#$ can also be obtained, but we will restrict attention to 
$p-$body interactions, that are somehow more natural.  
\end{rems}

In keeping with the previous remark, we focus on initial states $\rho$ in the sample that are even, {\it i.e.}  such that $\rho(a^{\#}_{k_1}\dots a^{\#}_{k_s})=0$ for all $s$ odd. 
\begin{thm}\label{evenDM}
Assume that the initial density matrix on $\F_-(\cH_s)$, $\rho^{even}$ is an even state, then,
the reduced density matrix on  $\F_-(\cH_s)$ at time $md$, denoted by $\rho^{even}_m$, is an even state, for all $m\geq 1$.
\end{thm}
\begin{proof}
The statement is a consequence of Thm \ref{thm:exact} and of the fact that $\omega_\Sigma$ is even, being quasifree. Indeed, 
for $s$ odd, 
\be
\rho^{even}_m(a^\#_{j_1} a^\#_{j_2}\cdots a^\#_{j_s})=(\omega_\Sigma \otimes \rho^{even}) \circ \tau^{md}(a^\#_{j_1} a^\#_{j_2}\cdots a^\#_{j_s}),
\ee
where $\tau^{md} (a^\#_{j_1} a^\#_{j_2}\cdots a^\#_{j_s})$ is a linear combination of products of an odd number of $a^\#$ and $b^\#$. The action of 
$\omega_\Sigma \otimes \rho^{even}$ on such products thus yields zero.\ep
\end{proof}

Let us introduce for all $s\in\Z$, the $d\times d$ matrix $\sigma^{(s)}$ whose entries are 
\be\label{defsigs}
\sigma^{(s)}_{j,k}=\Sigma_{k+sd,j}= \sigma(k-j+sd),
\ee
and which will play an important role below.

\begin{thm}\label{thm:1DM}
The one-body reduced density matrix on $\F_-(\cH_s)$ at time $md$ is given for any initial density matrix on the sample by 
\be
\rho^{(1)}_m=\cos^{2m}(\alpha) \rho^{(1)}+ (1-\cos^{2m}(\alpha))\sigma^{(0)}+\sum_{u=1}^{m-1} (\cos^{u}(\alpha)-\cos^{2m-u})(\sigma^{(u)}+\sigma^{(-u)}),
\ee
and satisfies the evolution equation
\be
\rho^{(1)}(m+1)=\cos^{2}(\alpha) \rho^{(1)}_m+ \sin^{2}(\alpha)B(m)
\ee
where 
\be
B(m)=\sigma^{(0)}+\sum_{u=1}^{m} \cos^{u}(\alpha)(\sigma^{(u)}+\sigma^{(-u)}).
\ee
For all $\alpha\notin \{0,\pi\}$,
\be\label{1dminfinite}
\rho^{(1)}_\infty=\lim_{m\to\infty}B(m).
\ee
 \end{thm}

\begin{rem}
As is true for all one-body reduced density matrices, ${\mathbb O}\leq \rho^{(1)}_\infty\leq \un$.
\end{rem}
\begin{proof}
Taking the expectation of  \eqref{evolova*a}  with respect to the quasi- free state $\omega_\Sigma$ yields for $m\geq 1$
\be\label{aaa}
\omega_\Sigma(\tau^{md}( a_k^* a_j))=\cos^{2m}(\alpha)  a_k^* a_j+ \sin^2(\alpha) \sum_{r=0}^{m-1}\sum_{s=0}^{m-1}(\cos(\alpha))^{2(m-1)-(r+s)} \sigma(k-j+(r-s)d).
\ee
In case $\cos(\alpha)=0$, this yields $\omega_\Sigma(\tau^{md}( a_k^* a_j))=\sigma(k-j)$. If $\cos(\alpha)\neq 0$,
the change of variables $u= s-r$  show that 
\bea
 \sum_{r=0}^{m-1}\sum_{s=0}^{m-1}\cos^{-(r+s)}(\alpha) \sigma(k-j+(r-s)d)= \sum_{u=0}^{m-1}\cos^{-u}(\alpha) \sigma(k-j-ud)\sum_{r=0}^{m-1-u} \cos^{-2r}(\alpha)\nonumber
\\+\sum_{u=-(m-1)}^{-1}\cos^{-u}(\alpha) \sigma(k-j-ud)\sum_{r=-u}^{m-1} \cos^{-2r}(\alpha).
\eea
Inserting this in the expectation, we get in all cases
\bea\label{expevola*a}
\omega_\Sigma(\tau^{md}( a_k^* a_j))=\cos^{2m}(\alpha)  a_k^* a_j+ 
(1-\cos^{2m}(\alpha))\sigma(k-j)\notag
\\+\sum_{u=1}^{m-1} (\cos^{u}(\alpha)-\cos^{2m-u}(\alpha))\big(\sigma(k-j-ud)+\sigma(k-j+ud)\big).
\eea
This along with definitions (\ref{denistyreduced}) and (\ref{defsigs}) give the first assertion of the theorem.  
The other two assertions are straightforward consequences, using that $\sup_{u\in\Z}\|\sigma^{(u)}\|<\infty$ for the existence of the limit.
\ep\end{proof}


The evolution of the two-body reduced density matrices can also be characterized, as well as the asymptotic evolution of all $p$-body matrices:
\begin{thm}\label{thm:2DM}
The two-body reduced density matrix at time $md$ is given for any initial density matrix on the sample  by 
\begin{align}\label{evoloved2dm}
\rho^{(2)}_m&=\cos^{4m}(\alpha) \rho^{(2)}+ 2\sin^2(\alpha)\cos^{2m}(\alpha)\P^{(2)}_A \Big(\rho^{(1)}\otimes \sum_{r,s=0}^{m-1} (\cos(\alpha)^{2(m-1)-(r+s)}\sigma^{(r-s)}\Big)\P^{(2)}_A\nonumber
\\&+ \sin^{4}(\alpha)\P^{(2)}_A \Big(\sum_{r,s=0}^{m-1} (\cos(\alpha)^{2(m-1)-(r+s)}\sigma^{(r-s)}\otimes \sum_{r,s=0}^{m-1} (\cos(\alpha)^{2(m-1)-(r+s)}\sigma^{(r-s)}\Big)\P^{(2)}_A,
\end{align}
where $\P^{(2)}_A$ is the projection of $\cH_s\otimes \cH_s$ into $\cH_s\wedge \cH_s$.

The two-body reduced density matrix satisfies the evolution equation
\begin{align}\label{evol2body}
\rho^{(2)}(m+1)&=\cos^{4}(\alpha) \rho^{(2)}_m+ \sin^{2}(\alpha)\P_A^{(2)}\big\{\cos^2(\alpha)\big(B(m)\otimes \rho^{(1)}_m\nonumber
\\&+\rho^{(1)}_m\otimes B(m)\big)+ \sin^2(\alpha) B(m)\otimes B(m)\big\}\P_A^{(2)}.
\end{align}

 For all $\alpha\notin \{0,\pi\}$ and all $1\leq p\leq d$, the $p$-body reduced density matrix acting on $\cH^{\wedge p}$ is
\be
\rho^{(p)}_\infty=\P^{(p)}_A \rho^{(1)}_\infty\otimes \dots \otimes \rho^{(1)}_\infty \P^{(p)}_A.
\ee 
 \end{thm}

\begin{rem}
The reduced asymptotic $p$-body density operator $\rho^{(p)}$ on $\cH^{\wedge p}_s$, such that 
\\$\rho^{(p)}=\P^{(p)}_A (\rho^{(1)}\otimes..\otimes \rho^{(1)})\P^{(p)}_A$ is a quasi-free state with symbol $\rho^{(1)}$ such that $\|\rho^{(p)}\|\leq 1$.
\end{rem}
\begin{proof}
For all $k_1\neq k_2$ and $j_1\neq j_2$, we have 
\be
\tau^{md}(a_{k_1}^*a_{k_2}^*a_{j_2}a_{j_1})=\big(U^{*md}(a_{k_1}^*a_{k_2}^*)U^{md}\big)\big(U^{*md}(a_{j_2}a_{j_1})U^{md}\big).
\ee
Using \eqref{a*evolvedk*}, \eqref{aevolvedk},
and taking the expectation with respect to the quasi-free state $\omega_\Sigma$, we get

\begin{align}\label{intermed}
\omega_\Sigma\big(\tau^{md}&(a_{k_1}^*a_{k_2}^*a_{j_2}a_{j_1})\big) = \cos^{4m}(\alpha)a_{k_1}^*a_{k_2}^*a_{j_2}a_{j_1}
 \\&-
\cos^{2m}(\alpha)a_{k_1}^*a_{j_2}\sin^2(\alpha)\sum_{r,s=0}^{m-1}\cos(\alpha)^{2(m-1)-(r+s)}\sigma^{(r-s)}_{j_1,k_2,}\nonumber
\\
&+\cos^{2m}(\alpha)a_{k_1}^*a_{j_1}\sin^2(\alpha)\sum_{r,s=0}^{m-1}\cos(\alpha)^{2(m-1)-(r+s)}\sigma^{(r-s)}_{j_2,k_2}\nonumber
\\&+\cos^{2m}(\alpha)a^*_{k_2}a_{j_2}\sin^2(\alpha)\sum_{r,s=0}^{m-1}\cos(\alpha)^{2(m-1)-(r+s)}\sigma^{(r-s)}_{j_1,k_1}\nonumber
\\&-\cos^{2m}(\alpha)a_{k_2}^*a_{j_1}\sin^2(\alpha)\sum_{r,s=0}^{m-1}\cos(\alpha)^{2(m-1)-(r+s)}\sigma^{(r-s)}_{j_2,k_1}\nonumber
\\&+\sin^4(\alpha)\sum_{r_1,s_1,r_2,s_2=0}^{m-1}\cos(\alpha)^{4(m-1)-(r_1+s_1+r_2+s_2)}\big(\sigma^{(r_1-s_1)}_{j_1,k_1}\sigma^{(r_2-s_2)}_{j_2,k_2}-\sigma^{(r_1-s_2)}_{j_1,k_2}\sigma^{(r_2-s_1)}_{j_2,k_1}\big),\nonumber
\end{align}
with the definition (\ref{defsigs}) of $\sigma^{(u)}$ on $\cH_s$.

Using the definition \eqref{defproj} of the projection $\P^{(2)}_A$, one sees that for operators $A,B$ on $\cH_s$
\be\label{p2abp2}
 \P^{(2)}_A (A\otimes B)\P^{(2)}_A= \P^{(2)}_A (B\otimes A) \P^{(2)}_A =\frac12 C,
\ee

where the operator $C$ on $\cH_s\wedge\cH_s$, is defined through its  matrix elements  
\be
C_{j_1j_2k_1k_2}= A_{j_1k_1}B_{j_2k_2}-A_{j_2k_1}B_{j_1k_2}+ A_{j_2k_2}B_{j_1k_1}-A_{j_1k_2}B_{j_2k_1}.
\ee
Hence, upon relabelling, the last term of (\ref{intermed}) reads
\be
\sin^4(\alpha)\P^{(2)}_A (\sum_{r,s=0}^{m-1}\cos^{2(m-1)-(r+s)}(\alpha)\sigma^{(r-s)}\otimes \sum_{r,s=0}^{m-1}\cos^{2(m-1)-(r+s)}(\alpha)\sigma^{(r-s)})\P^{(2)}_A.
\ee
Taking expectation with respect to the initial density in the sample, making use of definition (\ref{denistyreduced}), of (\ref{p2abp2}) again for the other sums in (\ref{intermed}), eventually yields \eqref{evoloved2dm}. 
The evolution equation (\ref{evol2body}) is a straightforward consequence of that. 

In order to compute the long time limit of the $p$-body  reduced density matrix, we use the properties of $\omega_\Sigma$ to rewrite \eqref{pdminfinite} in terms of the operators $\sigma^{(u)}$:
\begin{align}
\lim_{m\to\infty}& \omega _\Sigma\otimes\rho(\tau^{md}( a_{k_1}^* a_{k_2}^*\cdots a_{k_p}^* a_{j_p} \cdots a_{j_2} a_{j_1})) \\
&=\lim_{m\to\infty}\sin^{2p}(\alpha) \hspace{-.5cm}\sum_{r_1, ..,r_p, s_1,..,s_p=0}^{m-1} \hspace{-.5cm}\cos(\alpha)^{2p(m-1)-(r_1+s_1+\dots+r_p+s_p)}\nonumber
\sum_{\pi\in S_p} \epsilon_\pi\sigma^{(r_{\pi(1)}-s_1)}_{j_1, k_{\pi(1)}}...\sigma^{(r_{\pi(p)}-s_p)}_{j_p, k_{\pi(p)}}\notag
\\&=\lim_{m\to\infty} \sum_{\pi\in S_p} \epsilon_\pi\sin^{2p}(\alpha)\hspace{-.5cm} \sum_{r_1, ..,r_p, s_1,..,s_p=0}^{m-1} \hspace{-.5cm}\cos(\alpha)^{2p(m-1)-(r_1+s_1+\dots+r_p+s_p)}
 \sigma^{(r_{1}-s_1)}_{j_1, k_{\pi(1)}}...\sigma^{(r_{p}-s_p)}_{j_p, k_{\pi(p)}}\notag
\\&=\bra e_{j_1}\wedge\cdots \wedge e_{j_p}|\P_A^{(p)} \rho^{(1)}_\infty\otimes \dots \otimes \rho^{(1)}_\infty\P_A^{(p)}|e_{k_1}\wedge \cdots \wedge e_{k_p}\ket.
\nonumber
\end{align}

Here we used \eqref{1dminfinite} and the fact that 
\begin{align}
B(m)=&
\sin^{2}(\alpha) \sum_{r,s=0}^{m}\cos(\alpha)^{2m-(r+s)}
 \sigma^{(r-s)}\\ \nonumber &+\cos^{m+1}(\alpha)\left(\sigma^{(0)}\cos^{m+1}(\alpha)+\sum_{u=1}^{m}(\sigma^{(u)}+\sigma^{(-u)})\cos^{m+1-u}(\alpha)\right),
\end{align}
which follows from (\ref{aaa}) and (\ref{expevola*a}).
\ep
\end{proof}

As a direct application of the previous general results,  the following corollary, gives the long time limit of the expectation of the occupation number and the correlations in the  sample.
\begin{cor}\label{correlations}
For all $\alpha\notin \{0,\pi\}$  and all $j\in\{0,..,d-1\}$, the asymptotic expectation of the number operator in the sample at site $j$, $n^s_j$, is constant and  given by
\be \bra n^s_j\ket_{\rho_\infty}=
\lim_{m\to\infty}\omega_\Sigma\otimes \rho(\tau^{md} ( n^s_j))=\sigma(0) +\sum_{u=1}^{\infty} 2\Re(\sigma(ud))\cos^{u}(\alpha)=\bra e_0|\rho^{(1)}_\infty e_0\ket.
\ee 
The asymptotic expectation of correlations at sites $j\neq k\in \{0,..,d-1\}$ is given by
\be
\bra n^s_jn_k^s\ket_{\rho_\infty}=\lim_{m\to\infty}\omega_\Sigma\otimes \rho(\tau^{md}(  n^s_jn_k^s))=\bra e_0|\rho^{(1)}_\infty e_0\ket^2-|\bra e_0|\rho^{(1)}_\infty e_{k-j}\ket|^2.
\ee 
\end{cor}

\begin{rem}
We also deduce that
$\lim_{d\ra\infty}\bra N^s\ket_{\rho_\infty}/d=\sigma(0)$, where $N^s=\sum_{j=0}^d n_j^s$ is the total number of particles in the sample. Hence the asymptotic expectation of the  particle density in the sample coincides with $\sigma(0)$, the particle density in the reservoir, in the limit of large samples.
\end{rem}
Quasifree density matrices can be completely described by the set of all reduced $p$-body density matrices, see {\it e.g.}  \cite{DFP}. Hence Theorem \ref{thm:2DM} will allow us to compute explicitly $\rho(\infty)$ in case the initial state in the sample is even. We use the notation $\dfrac{\rho^{(1)}_\infty}{\idtyty -\rho^{(1)}_\infty}={\rho^{(1)}_\infty}{(\idtyty -\rho^{(1)}_\infty)^{-1}}$.


\begin{thm}\label{qfeven}
Assume that the initial density matrix on $\F_-(\cH_s)$, $\rho^{even}$, is an even state and that $|\cos(\alpha)|<1$. Then $\rho^{even}_\infty$ is a gauge invariant quasi-free state  
given on \\ $\cH_s\oplus \cH_s^{\wedge 2}\oplus...\oplus\cH_s^{\wedge d}$ by  

\begin{align}\label{eq:reductototal}
\rho^{even}_\infty&=\det(\idtyty -\rho^{(1)}_\infty)\times \\
&\times \left[\idtyty\oplus\dfrac{\rho^{(1)}_\infty}{\idtyty -\rho^{(1)}_\infty}\oplus \left(\dfrac{\rho^{(1)}_\infty}{\idtyty -\rho^{(1)}_\infty}\wedge\dfrac{\rho^{(1)}_\infty}{\idtyty -\rho^{(1)}_\infty}\right)\oplus \cdots\oplus \det\left(\dfrac{\rho^{(1)}_\infty}{\idtyty -\rho^{(1)}_\infty}\right)\right]\nonumber\\
&= \det(\idtyty -\rho^{(1)}_\infty)\exp\left\{d\Gamma \ln\left(\dfrac{\rho^{(1)}_\infty}{\idtyty -\rho^{(1)}_\infty}\right)\right\}.
\nonumber
\end{align}
\end{thm}

\begin{rems}
i)  The last formula shows that $\rho^{even}_\infty$ is a maximizer of $S(\rho)=-\tr_{\cF_-}(\rho \ln \rho)$ under the constraint 
$\left\langle d\Gamma \ln\left(\dfrac{\rho^{(1)}_\infty}{\idtyty -\rho^{(1)}_\infty}\right)\right\rangle_\rho$ fixed.\\
ii) The statement requires first ${\mathbb O}< \rho^{(1)}_\infty<\un$, and is then shown to hold by continuity as well if $\rho^{(1)}_\infty$ has non trivial kernel, or equals a projector, see \cite{DFP} for details.

\end{rems}

\begin{proof} We first show that $\rho^{even}_\infty$ has a block diagonal representation with respect to the subspaces $\cH^{\wedge p}$, $0\leq p\leq d$.  Note that for any  $s\geq 1$,
\be
\tau^{md}( a^{*}_{k_1}..a^{*}_{k_r}a_{j_1}..a_{j_r}a_{l_1}..a_{l_s})=\tau^{md}( a^{*}_{k_1}..a^{*}_{k_r}a_{j_1}..a_{j_r}) \tau^{md}(a_{l_1}..a_{l_s}).
\ee
If $s$ even, it is easy to see that 
\begin{align}\label{eq:evolvevena}
 &\tau^{md}(a_{l_1}..a_{l_s}) =\\ &\big(K^*_{l_1}..K^*_{l_1+md} a_{l_1}K^*_{l_1+md}..K^*_{l_1}\big)\big(K_{l_2}..K_{l_2+md} a_{l_2}K_{l_2+md}..K_{l_1}\big)....\big(K_{l_s}..K_{l_1+md} a_{l_s}K_{l_s+md}..K_{l_s}\big).\nonumber
\end{align}

Using \eqref{aevolvedk} and \eqref{a*evolvedk*} along with theorem \ref{thm:exact} and the fact that $\omega_\Sigma$ is gauge invariant quasi-free, proves that
\be\label{eq:limiteven}
\lim_{m\to\infty} \omega_\Sigma\otimes \rho^{even}( \tau^{md}( a^{*}_{k_1}..a^{*}_{k_r}a_{j_1}..a_{j_r}a_{l_1}..a_{l_s}))=0.
\ee
For $s$ odd, Theorem \ref{evenDM} directly implies \eqref{eq:limiteven}. Thus, $\rho^{even}_\infty$ is completely characterized by its restrictions to $\cH^{\wedge p}$, which are given by $\P^{(p)}_A \rho^{(1)}_\infty\otimes \dots \otimes \rho^{(1)}_\infty \P^{(p)}_A$, see Theorem \ref{thm:2DM}. From there on, the explicit form of $\rho^{even}_\infty$ follows from
 \cite{DFP} lemma 3; see also \cite{Derezinski}. \ep
\end{proof}

%
\section{Repeated Interaction System Case}\label{ris}
%
In the special case $\Sigma=\sigma\idtyty_{\cH_r}$,  with $\sigma\geq 0$ more can be deduced, since  the function $\sigma(k)$ take a simpler form with  $\sigma(k)=\sigma\delta_{0,k}$. Moreover, since this case corresponds to suppressing all correlations between the particles of the reservoir, the dynamics experienced by the particles in the sample corresponds to that of a repeated interaction system; the state of the particle in the reservoir the sample interacts with is always the same, irrespective of the time step, and after the interaction, it cannot influence the dynamics of the sample anymore. This defines a repeated interaction system, see {\it e.g.} the review \cite{BJM}, that we study in this and the following sections. In this section we still consider the exactly solvable case where the free dynamics in the sample is given by the shift $S_p$, while in the next section we deal with the general case of an arbitrary free dynamics in the sample. From a technical perspective, when $\Sigma=\sigma\idtyty_{\cH_r}$, the gauge invariant state $\omega:=\omega_{\sigma\idtyty_{\cH_r}}$ satisfies $\omega(b_{j}^*b_{k})=\delta_{jk}\sigma$, which we will use repeatedly. \\

Specializing to $\Sigma=\sigma\idtyty_{\cH_r}$,  we can complement the results of theorems \ref{thm:1DM} and \ref{thm:2DM} as follows:
\begin{prop}\label{tdpbdenmat}
For $\Sigma=\sigma\idtyty$, the $p$-body reduced density matrices at time $md$ take the forms
\begin{align} \label{eq:1reducedid}
\rho^{(p)}_m=
\sum_{0\leq k \leq p}{p \choose k}(1-\cos^{2m}(\alpha))^{p-k} \cos^{2mk}(\alpha)\P^{(p)}_A\big(\rho^{(k)}\otimes \sigma\idtyty_{\cH_s}\otimes...\otimes \sigma\idtyty_{\cH_s}\big)\P^{(p)}_A,
 \end{align}
for all  $p\in\{1,\dots, d\}$, where $\rho$ is the initial density matrix on the sample. 

\end{prop}
{\bf Proof:} 
We need to compute for $\{k_1<k_2<\dots <k_p\}$ and $\{j_1<j_2<\dots j_p\}$, 
\be
\bra e_{k_1}\wedge e_{k_2}\wedge\dots \wedge e_{k_p} | \rho^{(p)}_m e_{j_1}\wedge e_{j_2}\wedge \dots \wedge e_{j_p}\ket =\omega(\tau^{md}( a_{k_1}^* a_{k_2}^*\dots a_{k_p}^* a_{j_p} \dots a_{j_2} a_{j_1})).
\ee
Using (\ref{yyy}), (\ref{a*evolvedk*}) and (\ref{aevolvedk}) again to express $\tau^{md}( a_{k_1}^* a_{k_2}^*\dots a_{k_p}^* a_{j_p} \dots a_{j_2} a_{j_1})$, we can identity the conditions on the terms in the expanded product that do not vanish after taking expectation with respect to $\omega$: 
Each creation operator $b^*_{k_s}(r)$ stemming form $a_{k_s}^*$ must be paired  with an annihilation operator $b_{j_{s'}}(r')$ stemming $a_{j_{s'}}$, with $k_s=j_{s'}$, and $r=r'$. 

We say that the indices $k_s$  and $j_{s'}$ can be paired in this case. If $k_s$ and $j_{s'}$ can be paired and $k_{\tilde s}$ and $j_{\tilde s'}$ can be paired,  then $s<\tilde s$ and $s'<\tilde s'$, due to the ordering of the $k_s$'s and $j_s'$'s. Hence, there is a unique way to pair indices. Moreover, a term involving $u$ pairs of the previous form yields a factor $\sigma^u$ after expectation with respect to $\omega$. 

Therefore, the contribution from the pairing of indices $k_s$ and $j_{s'}$ stemming from $a^*_{k_s}$ and $a_{j_{s'}}$, further taking into account the respective factors $(-1)^s g_\alpha$ and $(-1)^{s'+1}g_\alpha$, yields a factor
\be
(-1)^{s+s'}\sin^2(\alpha)\frac{1-\cos^{2m}(\alpha)}{1-\cos^{2}(\alpha)}\sigma=(-1)^{s+s'}(1-\cos^{2m}(\alpha))\sigma.
\ee
In turn, if $\{k_1<k_2<\dots <k_p\}$ and $\{j_1<j_2<\dots j_p\}$ contain $u$ pairings of indices, {\it i.e.} 
$1\leq s_1<s_2<\dots <s_u\leq p$ and $1\leq s'_1<s'_2<\dots <s'_u\leq p$ such that $k_{s_\omega}=j_{s'_\omega}$, $\omega=1,\dots ,u$, the corresponding contribution after expectation with respect to $\omega$ equals
\be
(-1)^{\sum_{\omega=1}^u(s_\omega+s'_\omega)}(1-\cos^{2m}(\alpha))^u\sigma^u.
\ee
The contribution from the remaining terms that contain $a^*$ and $a$'s is simply 
\be
\cos^{2m(p-u)}\rho(a_{k_{s_{u+1}}}^* a_{k_{s_{u+2}}}^*\dots a_{k_{s_p}}^* a_{j_{s'_{p}}} \dots a_{j_{s'_{u+2}}} a_{j_{s'_{u+1}}}),
\ee
where $1\leq s_{u+1}<s_{u+2}<\dots <s_{p}\leq p$ and $1\leq s'_{u+1}<s'_{u+2}<\dots <s'_{p}\leq p$ are the set of indices that have no match for pairing, 
{\it i.e.} such that $\{k_{s_\omega}\}_{\omega=u+1\dots p}$ and $\{j_{s'_\omega}\}_{\omega=u+1\dots p}$ are distinct. Note that this contribution is proportional to a matrix element of $\rho^{(p-u)}$.

Let us compute the following matrix element, where the identity appears $u$ times, and where $\rho^{(p-u)}$ is viewed as a matrix acting on (the antisymmetric subspace of) $\cH^{\otimes (p-u)}$,
\begin{align}
&\bra e_{k_1}\wedge e_{k_2}\wedge\dots \wedge e_{k_p} | \cP^{(p)}_A \idtyty\otimes \idtyty \otimes \dots \otimes \idtyty \otimes \rho^{(p-u)} \cP^{(p)}_A e_{j_1}\wedge e_{j_2}\wedge \dots \wedge e_{j_p}\ket =  \\
&\sum_{\pi, \pi'\in{\mathfrak S}_p}\frac{\epsilon_\pi\epsilon_{\pi'}}{p!}\bra  e_{k_{\pi(1)}}\otimes e_{k_{\pi(2)}}\otimes\dots \otimes e_{k_{\pi(p)}} | \idtyty\otimes  \dots \otimes \idtyty \otimes \rho^{(p-u)} e_{j_{\pi'(1)}}\otimes e_{j_{\pi'(2)}}\otimes\dots \otimes e_{j_{\pi'(p)}}\ket = \nonumber \\
&\sum_{\pi, \pi'\in{\mathfrak S}_p}\frac{\epsilon_\pi\epsilon_{\pi'}}{p!}\delta_{k_{\pi(1)}, j_{\pi'(1)}}\dots \delta_{k_{\pi(u)}, j_{\pi'(u)}}
\bra  e_{k_{\pi(u+1)}}\otimes\dots \otimes e_{k_{\pi(p)}} | \rho^{(p-u)}  e_{j_{\pi'(u+1)}}\otimes\dots \otimes e_{j_{\pi'(p)}}\ket.\nonumber
\end{align}
Introducing the permutations $\theta$ and $\theta'$ by $\theta(\omega)=s_\omega$ and $\theta'(\omega)= s'_\omega$, so that 
$ \theta\circ \pi (\omega)=s_{\pi(\omega)}$ and similarly for primed quantities, the above can be rewritten as
\begin{align}\label{concat}
\sum_{\pi, \pi'\in{\mathfrak S}_p}\frac{\epsilon_{\theta\circ\pi}\epsilon_{\theta'\circ \pi'}}{p!}\delta_{k_{s_{\pi(1)}}, j_{s'_{\pi'(1)}}}\hspace{-.3cm}\dots \delta_{k_{s_{\pi(u)}}, j_{s'_{\pi'(u)}}}
\bra  e_{k_{s_{\pi(u+1)}}}\otimes\dots \otimes e_{k_{s_{\pi(p)}}} | \rho^{(p-u)}  e_{j_{s'_{\pi'(u+1)}}}\hspace{-.3cm}\otimes\dots \otimes e_{j_{s'_{\pi'(p)}}}\ket.
\end{align}
Now, by construction, the summand is possibly non zero only if $s_{\pi(\omega)}=s'_{\pi'(\omega)}$, {\it i.e.} only if  ${\pi(\omega)}={\pi'(\omega)}$, for $1\leq \omega\leq u$. Hence the summation can be reduced to permutations $\pi$  which consist in concatenations of the form 
$\pi=\pi_I\times\pi_D$, where $\pi_I\in {\mathfrak S}_u$ and $\pi_D\in {\mathfrak S}_{p-u}$ are permutations on $\{1,2,\dots, u\}$, respectively on $\{u+1,u+2,\dots, p\}$. Similarly, $\pi'=\pi'_I\times\pi'_D$, with the constraint $\pi_I=\pi'_I$. Hence, noting that $\epsilon_{\pi_I\times\pi_D}=\epsilon_{\pi_I}\epsilon_{\pi_D}$, (\ref{concat}) reads
\begin{align}
\epsilon_{\theta}\epsilon_{\theta'}&\frac{u!(p-u)!}{p!}\hspace{-.2cm}\sum_{\pi_D, \pi'_D\in{\mathfrak S}_{p-u}}\frac{\epsilon_{\pi_D}\epsilon_{\pi'_D}}{(p-u)!}\bra  e_{k_{s_{\pi_D(u+1)}}}\otimes\dots \otimes e_{k_{s_{\pi_D(p)}}} | \rho^{(p-u)}  e_{j_{s'_{\pi'_D(u+1)}}}\hspace{-.3cm}\otimes\dots \otimes e_{j_{s'_{\pi'_D(p)}}}\ket=\nonumber \\
&\epsilon_{\theta}\epsilon_{\theta'}\frac{u!(p-u)!}{p!}\bra e_{k_{s_{u+1}}}\wedge e_{k_{s_{u+2}}}\wedge\dots \wedge e_{k_{s_{p}}} | \rho^{(p-u)}
 e_{j_{s'_{u+1}}}\wedge e_{j_{s'_{u+2}}}\wedge\dots \wedge e_{j_{s'_{p}}} \ket.
\end{align}
Using the defintions of $\theta$ and $\theta'$, one easily sees that $\epsilon_{\theta}\epsilon_{\theta'}=(-1)^{\sum_{\omega=1}^u (s_\omega+s'_\omega)}$. Thus, we get that the matrix elements of 
$\rho^{(p)}_m$ between vectors that contain $u$ pairings of indices, coincide with the summand in (\ref{eq:1reducedid}) with $k=p-u$. Taking into account all possible pairings yields the result.
\ep

\medskip 

Consequently, the expectation of the number operator and the correlation at any time is given by the following lemma 
\begin{lem}\label{lem:identity-number}
For $\Sigma=\sigma\idtyty_{\cH_r}$ and $j,r\in\{0,..,d-1\}$, the time evolution  of the number $n^s_j$ as operators on $\F_-(\cH_s)$ is given by
\be\label{eq:identity-number}
\omega(\tau^{md+r} ( n^s_j))=\begin{cases} \cos^{2m}(\alpha)n^s_{ j+r} +\big(1-\cos^{2m}(\alpha)\big)\sigma, &\text{if  } j+r \leq d-1
\\\cos^{2(m+1)}(\alpha)n_{k}^s+\big(1-\cos^{2(m+1)}(\alpha)\big) \sigma   &\text{if  } j+r >d-1,
\end{cases}
\ee 
with  $ k=j+r-d$  
\end{lem}
\begin{proof}
This is a direct consequence of 
Remark \ref{remrem} and equation \eqref{expevola*a} along with the fact that  the function $\sigma(k)=\sigma\delta_{0,k}$.
\ep\end{proof}

This allows us to compute explicitly the full statistics of the number of particles in the sample for all times, when the sample is initially in the vacuum state:   
\begin{cor}\label{corrn=p}
If $\rho=|\Omega\ket\bra\Omega|$,  the probability of finding $p$ particles in the sample at time $md$, is given by 
\be
\mathbb{P}_m(N_s=p)={d \choose p}\big[(1-\cos^{2m}(\alpha))\sigma\big]^p\big[1-(1-\cos^{2m}(\alpha))\sigma\big]^{d-p},
\ee
{\it i.e.} the  number of particles in the sample at time $md$ is given by a binomial distribution 
$B(d,(1-\cos^{2m}(\alpha))\sigma)$.
\end{cor}
\begin{proof}
Let $P^s_p$ on $\cF_-(\cH_s)$ be the projector on the eigenspace associated to the eigenvalue $p$ of $N_s$, spanned by $\{a_{k_1}^*a_{k_2}^*\cdots a_{k_p}^*\Omega\}_{k_1<k_2<\cdots<k_p}$.  Using the identities
$
|\Omega\ket\bra\Omega|=\Pi_{j=0}^{d-1}a_ja^*_j=\Pi_{j=0}^{d-1}(\idtyty_{\cH_s}-n_j^s)
$, and $a_k^*(\idtyty_{\cH_s}-n_k^s)a_k=n_k^s$, we have
\be
P^s_p=\sum_{k_1<k_2<..<k_p} n^s_{k_1}...n^s_{k_p}\Pi_{j\notin\{k_1,..,k_p\}}(\idtyty_{\cH_s}-n^s_j).
\ee
Hence, the probability of finding exactly $p$ particles in the sample at time $md$, is given by
\be\label{def:probability}
\mathbb{P}_m(N_s=p)=\sum_{k_1<k_2<..<k_p} \omega\otimes |\Omega\ket\bra\Omega| \Big(\tau^{md}\big( n^s_{k_1}...n^s_{k_p}\Pi_{j\notin\{k_1,..,k_p\}}(\idtyty_s-n^s_j)\big)\Big).
\ee 
The evolution of products of number operators is addressed as in the proof of Theorem \ref{thm:exact}:
\begin{align}
\tau^{md}\big( n^s_{k_1}...n^s_{k_p}\Pi_{j\notin\{k_1,..,k_p\}}(\idtyty_s-n^s_j)\big)=
\tau^{md}(n^s_{k_1})...\tau^{md}(n^s_{k_p})\Pi_{j\notin\{k_1,..,k_p\}}\tau^{md}(\idtyty_s-n^s_j),
\end{align}
using (\ref{evolvns}), (\ref{a*evolvedk*}) and (\ref{aevolvedk}). The expectation of such products with respect to $\omega$ for $\Sigma=\sigma\idtyty_{\cH_r}$ shows that all products that involve a single operator $b^\#$ vanish, and we use Lemma \ref{lem:identity-number} for the other terms. More precisely, the time evolution of correlations for distinct set $A, B$ and all  $i \in A\subset\{0,..,d-1\}$, $j\in B\subset\{0,..,d-1\}$ is given by
\bea
\omega(\tau^{md}( \Pi_{i\in A}n^s_{i}\Pi_{j\in B}(\idtyty_{\cH_s}-n_{j}^s)))=\Pi_{i\in A} \big[\cos^{2m}(\alpha)n^s_{i} +(1-\cos^{2m}(\alpha))\sigma\big]\nonumber
\\\times \Pi_{j\in B}\big[\cos^{2m}(\alpha)(\idtyty-n^s_{j})+(1-\cos^{2m}(\alpha))(1-\sigma)\big].
\eea
Combining this with the definition of $\mathbb{P}_m(N_s=s)$ gives the result.
\ep\end{proof}

The long time limit of expectations in the sample, for arbitrary initial density matrix, are given in the next proposition.
\begin{prop}
For $\Sigma=\sigma\idtyty_{\cH_r}$ and $|\cos(\alpha)|<1$, the following is true, for any initial condition $\rho$:\\
i) The asymptotic $p$-body reduced density matrices are given for all $p\in\{1,\dots, d\}$ by
\be\label{r1infinity}
\rho^{(p)}_\infty=\sigma^p\idtyty_{\cH^{\wedge p}}.
\ee

ii) The asymptotic reduced density matrix on  $\F_-(\cH_s)$  is given by
\be
\rho_\infty= \frac{e^{-\mu N_s}}{Z(\mu)},
\ee
where $N_s$ is the number operator in the sample, 
and $\mu=\ln\frac{1-\sigma}{\sigma}$ and $Z(\mu)=(1-\sigma)^{-d}$.\\

iii) The total number of particles in the sample at time $md$ is given by a binomial distribution  $B(d,\sigma)$.

\end{prop}

\begin{rem}
It is worth noting that in this special case, the long time limit reduced density matrix is independent of the coupling parameter $\alpha$.
\end{rem}
\begin{proof}
The first statement is a direct consequence of \eqref{eq:1reducedid}.  
To prove the second statement, we start by showing that $\rho_\infty$ is diagonal in the basis $\{a^*_{k_1}\cdots a^*_{k_r}\Omega \}_{k_1<k_2<\dots <k_r}^{r=0,1,\dots, d-1}$: 
For distinct $j_1,..,j_s$, $s$ even and all $m\in\N$, using (the obvious generalization of) \eqref{eq:evolvevena}, \eqref{a*evolvedk*}, \eqref{aevolvedk}, 
and \eqref{defquasifree} with $\Sigma=\sigma\idtyty_{\cH_r}$,
\bea
\rho_\infty (a_{j_1}^{\#}a_{j_2}^{\#}...a_{j_s}^{\#})=\lim_{m\to\infty}\omega\otimes \rho\big(\tau^{md} (a_{j_1}^{\#}a_{j_2}^{\#}...a_{j_s}^{\#})\big)
=\lim_{m\to\infty}\cos^{ms}\rho( a_{j_1}^{\#}a_{j_2}^{\#}...a_{j_s}^{\#})=0.
\eea
To show that the same is true for $s$ odd,  we note that for all $j\in\{0,..,d-1\}$ and $m\in\N$
\be\label{eq:evolve1a}
\omega \big(\tau^{md}(a_{j}^{\#})\big)= \cos^{m}(\alpha)\omega \big(a_{j}^{\#} \Pi_{k\in P} K^2_{k}\big),
\ee
where $P\subset\{0,..,md-1\}\backslash\{j, j+d,..,j+md\}$, using Lemma \ref{evolkj} repeatedly and $\Sigma=\sigma\idtyty_{\cH_r}$. Now, for any $s$ odd,  it is enough to look at $\rho(\infty)(a_{j_1}^*..a^*_{j_r}a_{l_r}..a_{l_r} a^{\#}_k)$,  
\be\label{eq:evolveodd}
\omega\otimes \rho\big(\tau^{md}( a_{j_1}^*..a^*_{j_r}a_{l_r}..a_{l_r} a^{\#}_k)\big)= \omega\otimes\rho \big(\tau^{md}(a_{j_1}^*..a^*_{j_r}a_{l_r}..a_{l_r})\tau^{md}(a^{\#}_k)\big).
\ee
Since $|\omega\otimes \rho(B^*A)|^2\leq \big(\omega\otimes \rho(B^*B)\big(\omega\otimes\rho(A^*A)\big) $, combining  equations \eqref{eq:evolve1a}, \eqref{eq:evolveodd} and taking the limit as $m\to\infty$  implies $\rho_\infty$ is even. Using the expression of the basis vectors, as in the proof of Corollary \ref{corrn=p}, completes the proof that  $\rho_\infty$ is diagonal.
Then, the proof of theorem \ref{qfeven} applies with $(\ref{r1infinity})$ giving the required result.

Finally, the third statement is a direct consequence of Corollary \ref{corrn=p}, since the initial condition plays no role.\ep
\end{proof}
%
\subsection{Dynamics of the Flux}
The following proposition gives the time evolution of $\Phi_r$, whose definition (\ref{def:flux}) we recall 
\be
\Phi_r=\sin^2(\alpha)(\idtyty\otimes n^s_0-n_0^r\otimes \idtyty)+i\sin(\alpha)\cos(\alpha)(b_0^*\otimes a_0-b_0\otimes a_0^*),
\ee
 in the context of repeated interactions under consideration.
\begin{prop}\label{fluxsat}
For $\Sigma=\sigma\idtyty_{\cH_r}$, and for all $t=md+u$, $m\in\N$, $ u\in\{0,\dots, d-1\}$ we have 
\be
\omega(\tau^t(\Phi_r))=
\sin^2(\alpha)\cos^{2m}(\alpha)(n_u^s-\sigma).
\ee
Moreover for $t=md$,
\be\label{fltm}
\sum_{j=0}^{t-1}\omega\big(\tau^{j}(\Phi_r)\big)= (N_s-\sigma d)(1-\cos^{2m}(\alpha))=-\omega\big(\tau^{t}(N_s)\big)+N_s.
\ee
\end{prop}
\begin{rems}
 i) In particular, an application of Lemma \ref{lem:identity-number} yields for any initial state $\rho$
$
\omega\otimes \rho\big(\tau^{t}(\Phi_r)\big)=\sin^2(\alpha)\cos^{2m}(\alpha)\big(\rho(n_u^s)-\sigma\big),
$
which can have either sign. If $|\cos(\alpha)|<1$, $\lim_{t\ra\infty}\omega\big(\tau^{t}(\Phi_r)\big)\to 0$, a manifestation of thermalization. \\
ii) For $|\cos(\alpha)|<1$,  
\be
\lim_{m\ra\infty}\sum_{j=0}^{md-1}\omega\big(\tau^{j}(\Phi_r)\big)= (N_s-\sigma d), \ \ \ \mbox{and   }\ \ \lim_{m\ra\infty} \omega\big(\tau^{md}(N_s)\big)=\sigma d,
\ee
showing that, asymptotically, the total number of particles having entered the reservoir equals the difference of the initial number of particles in the reservoir with its asymptotic  value $\sigma d$, in keeping with the particle density in the reservoir.

\end{rems}
\begin{proof}
 For all $t\in\N$, $\Sigma=\sigma\idtyty_{\cH_r}$ implies the contribution of the evolution of the cross terms $\Phi_r$ vanish; moreover the evolution of 
$n_0^r$ is simply $n_t^r$, hence
\be
\omega(\tau^t(\Phi_r))=\sin^2(\alpha)\omega\big(\tau^t(n_0^s)-n^r_t\big)=\sin^2(\alpha)\big(\omega\big(\tau^t(n_0^s)\big)-\sigma\big).
\ee
Then, for $t-1=(m-1)d+u$, $m\in\N$, $u\in\{0,\dots, d-1\}$, using Lemma \ref{lem:identity-number}, we compute
\begin{align}
\sum_{j=0}^{t-1}\omega\big(\tau^{j}(\Phi_r)\big)
=(1-\cos^{2(m-1)}(\alpha))(N_s-\sigma d)+\sin^2(\alpha)\cos^{2(m-1)}(\alpha)((n_0^s+\dots+n_u^s)-\sigma (u+1)).
\end{align}
Specializing to $t=md$ with $u=d-1$, we get the first equality of (\ref{fltm}).
On the other hand, 
\be
\omega(\tau^{md}(N_s))=(N_s-\sigma d)(\cos^{2m}(\alpha)-1)+N_s,
\ee
which gives the second part of (\ref{fltm}).
\ep

\end{proof}

%
%
\section{General Dynamics in the Sample}\label{perdynsam}
%

In this section we still work in the framework of repeated interaction systems characterized by $\Sigma=\sigma\idtyty_{\cH_r}$. However, we return to the general case of arbitrary free dynamics on the sample given by \eqref{def:perturbeddyn}, 
so that $U_F=\Gamma(S)\otimes \Gamma (W)$, see \eqref{freedyn}, and again $K=K_0$ is given by \eqref{coupling}. This means that the one particle one time step evolution in the sample is determined by the unitary operator $W$ on $\cH_s$.

The main results of this section say that under natural assumptions stating that the one-body dynamics $W$ is mixing enough, which we express in terms of a spectral hypothesis, the long time asymptotics of the reduced one- and two-body density matrices in the sample exist and are  independent of the the sample dynamics and of the coupling strength. Moreover, they are given by $\sigma\idtyty_{\cH_s}$ and $\sigma^2\idtyty_{\cH_s^{\wedge 2}}$ respectively.

We start by some basic properties of the free dynamics in the sample $\Gamma(W)\equiv \idtyty\otimes \Gamma(W)$. 
In order to simplify the notation, we introduce the vectors and matrix
 \be
\vec{a^{\#}}=\begin{pmatrix}a^{\#}_0 &a^{\#}_1& \cdots & a^{\#}_{d-1} 
\end{pmatrix}^T \ \ \mbox{and} \ \ {\bf a^*a}^T=(a_k^*a_j)_{j,k\in\{0,1, \dots, d-1\}}.
\ee
Using Bogoliubov transform, one has
\begin{lem}\label{lem:aW} 
Let $(W)_{j,k}$ be the matrix representation of the operator $W$ in the canonical basis $\{e_0,\cdots, e_{d-1}\}$ of $\cH_s$.
Then, component-wise,
\bea
\Gamma(W)^*\vec{a}\Gamma(W)= W\vec{a}, \ \ \ \ \ \Gamma(W)^*\vec{a^{*}}\Gamma(W)= \overline W\vec{a^{*}}.
\eea
\end{lem}
\begin{proof}
Using the definition of $\Gamma(W)$ and the linearity of creation operators, we get 
\be
\Gamma(W)^* a_k^*\Gamma(W)= a^*(W^* e_k)= \sum_{j=0}^{d-1} \overline{W}_{kj} a_j^{*},
\ee
using
$
W^* e_k = \sum_{j=0}^{d-1} W^*_{jk}e_j. 
$
\ep
\end{proof}

We now determine the corresponding one-body density matrix for all times.
\begin{thm}
\label{thm:1DMpertub}  Let $\Sigma=\sigma\idtyty$ and $t\in\N$.
There exist $d\times d$ matrices $\M, \B$ such that the one-body reduced density matrix  $\rho^{(1)}_t$ at time $t$, with initial condition $\rho^{(1)}$, is given by
\be
\rho^{(1)}_t= \M^{t}\rho^{(1)} \M^{*t}+\sum_{r=0}^{t-1}\M^{r} \B \M^{*r} .
\ee
Equivalently, for all $t\in \N$
\be
\rho^{(1)}_{t+1}=\M \rho^{(1)}_t \M^{*} +\B , \ \ \ \rho^{(1)}_0=\rho^{(1)}.
\ee
The matrices $\M $ and $\B $  are defined as
\begin{align}\label{def:M}
\M &= W \K, \hspace{1.2cm}\mbox{with }\ \ \K=\idtyty+f_\alpha |e_0\ket\bra e_0 |, \ \ \mbox{and } \ f_\alpha=\cos(\alpha)-1,\\
\label{def:B}
\B &= W \Ee  W^*,  \hspace{0.7cm}\mbox{with }\ \ \Ee
 =\sigma\sin^2(\alpha)|e_0\ket\bra e_0|.
\end{align}
\end{thm}
\begin{rems}
i)  $\K=\K^*$ is diagonal. Also, in case $\rho^{(1)}_0=0$, $\B =\rho^{(1)}_1\geq 0$. \\
ii) The matrix $\M $ is of norm one, with singular values $\{\cos^2(\alpha), 1\}=\sigma(\K^2)$.\\
iii) In case $W$ is perturbation of $S_p$, since $\sigma(S_p\K)=\{e^{2i\pi k/d}\cos(\alpha)^{1/d}\}_{k=0,1,\dots, d-1}$,  then $\M $ has all its eigenvalue within the unit circle if $\alpha\notin\{0, \pi\}$ and
$W-S_p$ is small enough.\\
iv) If $\alpha\in\{0, \pi\}$, $\M $ is unitary and $\B =0$.\\
v) Setting $W=S_p$ and $t=md$, $m\in\N$, a straightforward computation shows that we recover Theorem \ref{thm:1DM} for $\Sigma=\sigma\idtyty_{\cH_r}$.
\end{rems}
\begin{proof}
Let us compute the evolution of $a^*_ka_j$ under $U_W=U_FK_0$. Using $[\Gamma(S),a^\#_k]=0$ Lemma \ref{lem:aW} and Lemma \ref{evolkj}, we get
\begin{align}\label{evpera*a}
U_W^*a^*_ka_j U_W=K_0^*\Gamma(W)^*  a^*_{k}a_{j} \Gamma(W)  K_0=\sum_{r,s}\overline{W}_{k r} W_{j s} 
K_0^*a^*_ra_sK_0,
\end{align}
where
\begin{align}\label{evlinper}
K_0^*a^\#_ra^\#_sK_0=\begin{cases} a^\#_r a_s^\# & \mbox{if }0\notin \{r,s\}
\\ (\cos(\alpha) a^\#_0-i\sin(\alpha) b_0^\#) a_s^\#&  \mbox{if }s\neq 0 = r 
\\ a_r^\# (\cos(\alpha) a_0^\#+i\sin(\alpha) b_0^\#) &  \mbox{if }r\neq 0 =s
\\ \cos^2(\alpha) a_0^\# a_0^\#+ \sin^2(\alpha) b_0^\#b_0^\#+i\sin(\alpha)\cos(\alpha)(b_0^\# a_0^\#-b_0^\#a_0^\#) &  \mbox{if }r=s=0.\end{cases}
\end{align}
Now that we have an expression for  $U_W^*a^*_ka_j U_W$, let us see how to  get hold of an arbitrary number of evolution steps. Using that for all $k\in\N$, $[\Gamma(W) ,b_k^\#]=0$,  and $K_0 b_k^\#=b_k^\#K_0^*$, $k\neq 0$, we compute, 
\begin{align}
\label{evprobb}
U_W^* b_k^*b_{k'}U_W&=K_0^*b_{k+1}^*b_{k'+1}K_0=b_{k+1}^*b_{k'+1}, \ \ \ \forall k, k' \in \N\\
\label{evprob}
U_W^* b_k^\#a_sU_W&=  b_{k+1}^\#K_0\Gamma(W)^* a_{s}\Gamma(W) K_0=b_{k+1}^\#\sum_{r}{W_{s r}}K_0a_rK_0,
\end{align}
where 
\be
K_0a^\#_rK_0=\begin{cases}a^\#_r & \mbox{if } r\neq 0 \\ \cos(\alpha) a^\#_0+i\sin(\alpha) b^\#_0 & \mbox{if } r= 0.\end{cases}
\ee

Observe that taking the expectation of $U_W^*a^*_ka_j U_W$ with respect to $\omega$ with $\Sigma=\sigma\idtyty_{\cH_r}$, makes the linear contributions in $b^\#_0$ vanish and replaces $b_0^*b_0$ by $\sigma$. Further exploiting the relations above and their adjoints, and the form of $\omega$, one sees that those linear terms in $b_0^\#$, when further evolved up to time $t$ yield a term $b_{t-1}^\#$ that cannot be paired with another $b^\#$ of same index and thus vanishes after expectation with respect to $\omega$. Similarly, the arbitrary evolution of $b_0^*b_0$ still gives rise to $\sigma$ after expectation with respect to $\omega$.
Thus, it follows by induction that we can discard the linear terms in $b_0^\#$ and replace $b_0^*b_0$ by its expectation, the scalar $\sigma$ in (\ref{evpera*a}) and (\ref{evlinper}) to get the evolution equation of $\omega (U^{*t}_W a_k^*a_j U^t_W)\equiv \omega (U^{*t}_W {\bf a^*a}^T U^t_W))_{kj}$.  Using the definitions (\ref{def:M}) and (\ref{def:B}) we have obtained
\begin{align}
\omega  (U^{*}_W a_k^*a_j U_W)=\sum_{r,s}(\overline{W} \K)_{kr}a^*_ra_s(W\K)_{js}+ \sigma\sin^2(\alpha)\overline W_{k0} W^T_{0j},
\end{align}
which, since $\K=\overline \K=\K^*$, takes the matrix form writes 
\be
\omega (U^{*}_W {\bf a^*a}^T U_W)= \overline{W \K} \, {\bf a^*a}^T\, (W\K)^T+\overline W \Ee W^T,
\ee
where,
\be
\overline W \Ee W^T=\B^T \ \ \mbox{and} \ \overline{W \K}=\overline{\M }.
\ee
As noted above, we can now iterate this relation to get for any $t\in\N$
\be
\omega (U^{*t}_W {\bf a^*a}^T U^t_W)=(\overline{\M })^t \, {\bf a^*a}^T\, (\overline{\M })^{*t} +\sum_{r=0}^{t-1}
(\overline{\M })^r \B^T  (\overline{\M })^{*r}.
\ee
In turn, we deduce the time evolution of (the transpose of) $\rho^{(1)}_t$, for any $t\in\N$, by applying $\rho$ to the above identity, which
ends the proof.
\ep
\end{proof}


In order to simplify the notation, let $\Mm$ be the contraction on the space of complex $d \times d$ matrices, defined as
\bea
\Mm:M_d(\C)\to M_d(\C),&&\Mm(A)= \M  A \M^* .
\eea

Using this notation, the results of Theorem \ref{thm:1DMpertub}, can be rewritten for $t\in\N$ as
\be\label{eq:1DMpertublinear}
\rho^{(1)}_t= \Mm^{t}(\rho^{(1)})+\sum_{r=0}^{t-1}\Mm^{r}(\B )\ \ \mbox{and}\ \
\rho^{(1)}_{t+1}= \Mm^{t}(\rho^{(1)}_t)+\B .
\ee

\begin{cor}\label{44}
Assume $\sigma(\Mm)\cap \Ss=\emptyset$. Then, for all $t\in\N$
\be
\rho^{(1)}_t= \Mm^t\big(\rho^{(1)}-(\idtyty-\Mm)^{-1}(\B )\big)+(\idtyty-\Mm)^{-1}(\B ). 
\ee
Moreover, 
\be\label{perturbed1dminfty}
\lim_{t\to\infty}\rho^{(1)}_t=(\idtyty-\Mm )^{-1}(\B )=\sigma\idtyty. 
\ee 
\end{cor}
Note the easily proven 
\begin{lem} For the contractions $\M$ and $\Mm$ we have
\be
\sigma(\M )\cap \Ss=\emptyset \ \ \Leftrightarrow \ \ \sigma(\Mm )\cap \Ss=\emptyset.
\ee
\end{lem}
\begin{rems}
i) The asymptotic one-body density matrix $\rho^{(1)}_\infty=\sigma\idtyty$ is independent of both $\alpha$ and the dynamics $W$, as long as the spectral condition holds. \\
ii) The spectral assumption does not hold if $|\cos(\alpha)|=1$, or if $W$ prevents some sites to be visited. \\
iii) If $W-S_p$ is small enough, the corollary holds, when $\alpha\notin\{0,\pi\}$.\\
iv) If the spectral assumption doesn't hold, the limit may exist but it depends on the initial state in general: Consider $W=\idtyty_{\cH_s}$, so that there is no dynamics in the sample, and $\alpha\notin\{0,\pi\}$. Then $\M =\K$ has spectrum $\{1,\cos(\alpha)\}$ and $\B =\Ee$. We easily compute 
\be
\sum_{r=0}^{t-1}\K^{r}\Ee \K=\sigma (1-\cos^{2t}(\alpha))|e_0\ket\bra e_0|\ra \sigma |e_0\ket\bra e_0| \ \ \mbox{as} \ \ t\ra \infty.
\ee
Similarly, $\lim_{t\ra \infty} \K^t \rho^{(1)}\K^t=P_0^\perp \rho^{(1)} P_0^\perp$, where $P_0^\perp=\idtyty-|e_0\ket\bra e_0|$, so that
\be
\lim_{t\to\infty}\rho^{(1)}=\sigma |e_0\ket\bra e_0|+ P_0^\perp \rho^{(1)} P_0^\perp.
\ee
\end{rems}

\begin{proof}[of Corollary \ref{44}]
First we note that our assumptions implies $ (\idtyty-\Mm^t )$ is invertible for any $t\in\N$, so that
\be
\sum_{r=0}^{t-1} \Mm^r =  (\idtyty-\Mm^t )(\idtyty-\Mm )^{-1}.
\ee
This combined with \eqref{eq:1DMpertublinear} gives the first assertion. The spectral assumption with a Jordan form argument allow us to take the limit $t\to\infty$ to get
\be
\rho^{(1)}_\infty= (\idtyty-\Mm )^{-1}(\B ). 
\ee
This identity is equivalent to 
\be
\rho^{(1)}_\infty=\M \rho^{(1)}_\infty\M^* +\B . 
\ee
Thus proving \eqref{perturbed1dminfty} amounts to proving that
\be\label{idtyty}
\idtyty=\M \M^* +\frac1\sigma\B .
\ee
Using the definitions (\ref{def:M}), (\ref{def:B}) together with the fact that $W$  is unitary, this is equivalent to proving
\be
\idtyty=\K^2+\frac1\sigma\Ee,
\ee
which is easily seen to hold true using the definitions of $K$ and $\Ee$.
\ep\end{proof}

Let us now turn to the dynamics of the 2-body reduced density matrix for an arbitrary dynamics in the sample, which contains informations on correlations. The same strategy as that applied to prove Theorem \ref{thm:1DMpertub} allows to show the following
\begin{thm}\label{perdyn2b}
Let $\Sigma=\sigma\idtyty$ and $t\in\N$.
There exist $d\times d$ matrices $ \M , \N , \B $ such that $\rho^{(2)}_t$ the two-body reduced density matrix at time $t$, with initial condition $\rho^{(2)}$, is determined by
\begin{align}\label{tdtbdm}
\rho^{(2)}_{t+1}=\cP_A^{(2)} \M ^{\otimes 2}\cP_A^{(2)} \rho^{(2)}_{t} \cP_A^{(2)} {\M^*}^{\otimes 2}\cP_A^{(2)}
+ 2\cP_A^{(2)} \B \otimes (\N \rho^{(1)}_t\N^* )\cP_A^{(2)},
\end{align}
or equivalently for $t\geq 1$ by
\begin{align}\label{rh2att}
\rho^{(2)}_{t}=&(\cP_A^{(2)} \M ^{\otimes 2}\cP_A^{(2)})^t \rho^{(2)}_{0} (\cP_A^{(2)} {\M^*}^{\otimes 2}\cP_A^{(2)})^t
\nonumber\\
&+ \sum_{r=0}^{t-1}2{(\cP_A^{(2)} \M ^{\otimes 2}\cP_A^{(2)})}^r\cP_A^{(2)} \B \otimes (\N \rho^{(1)}_{t-1-r}\N^* )\cP_A^{(2)}{(\cP_A^{(2)} {\M^*}^{\otimes 2}\cP_A^{(2)})}^r,
\end{align}
with $\rho^{(1)}_t$ the one-body density matrix, $\M $ and $\B $ as in Theorem \ref{thm:1DMpertub}, and $\N = W P^\perp$, where $P^\perp=\idtyty-|e_0\ket\bra e_0|$.
\end{thm}
\begin{rem}
Again, setting $W=S_p$ and $t=md$, $m\in\N$, we recover Theorem \ref{thm:2DM} for $\Sigma=\sigma\idtyty_{\cH_r}$ by explicit computation.
\end{rem}
\begin{proof}
We first compute   
\be \omega (U_W^*a_{k_1}^*a_{k_2}^*a_{j_2}a_{j_1} U_W))=\omega ((U_W^*a_{k_1}^*a_{k_2}^*U_W)( U_W^*a_{j_2}a_{j_1} U_W)),
\ee
and discuss the iteration of the dynamics, using $\Sigma=\sigma\idtyty_{\cH_r}$. Similarly to (\ref{evpera*a}), we have
\begin{align}
&(U_W^*a_{k_1}^*a_{k_2}^*U_W)( U_W^*a_{j_2}a_{j_1} U_W)=\\ \nonumber
&\sum_{{r_1, r_2}\atop {s_1, s_2}} \overline W _{k_1, r_1} \overline W_{k_2, r_2} W_{j_2, s_2} W_{j_1, s_1}
(K_0^*a^*_{r_1}a^*_{r_2}K_0)(K_0^*a_{s_2}a_{s_1}K_0),
\end{align}
We note that since $r_1=r_2$ or $s_1=s_2$ yield zero, there are either $0, 1$ or $2$ occurrences of indices equal to zero. Applying  (\ref{evlinper}) to the above and expanding the products, this yields terms with four operators $a^\#$ only, one operator $b_0^\#$ and three $a^\#$s, or a pair $b^*_0b_0$ and two $a^\#$s. The factors that are linear in $b_0^\#$ produce no contribution after expectation with respect to $\omega$.  By (\ref{evprob}) and (\ref{evprobb}),  iterating $t$ times the dynamics on those terms map the operator $b_0^\#$ to $b_{t-1}^\#$ that cannot be paired with other operators created in the bath which all have indices smaller than $t$. The pair $b_0^*b_0$ in the quadratic terms in $a^\#$ becomes a scalar factor $\sigma$ under $\omega$ and the same argument shows that under $t$ iterations of the dynamics, this pair becomes a factor $b_{t-1}^*b_{t-1}$ in front of the evolution of the quadratics terms in $a^\#$, which yields the same scalar factor $\sigma$ after applying $\omega$. This allows us to simply replace $b_0^*b_0$ by the number $\sigma$, and to get rid of terms in $b^\#$ which cannot be paired  in the iteration process. Altogether, with $\tau^t $ denoting the perturbed dynamics, we get with  $\K$ defined by (\ref{def:M}) to take into account the factors $\cos(\alpha)$ stemming from the zero indices,

\begin{align}\label{intermr2}
\omega (\tau^1 (a_{k_1}^*a_{k_2}^*a_{j_2}a_{j_1}))&=\sum_{{r_1\neq r_2}\atop {s_1\neq s_2}}\overline{W\K}_{k_1, r_1} \overline{W\K}_{k_2, r_2} (W\K)_{j_2, s_2} (W\K)_{j_1, s_1}a_{r_1}^*a_{r_2}^*a_{s_2}a_{s_1}\\ \nonumber
&+\overline{W}_{k_2, 0}W_{j_2, 0}\sum_{r_1\neq 0, s_1\neq 0 }\sigma\sin^2(\alpha)\overline{W}_{k_1, r_1} W_{j_1, s_1}a^*_{r_1}a_{s_1}\\ \nonumber
&+\overline{W}_{k_1, 0}W_{j_1, 0}\sum_{r_2\neq 0, s_2\neq 0 }\sigma\sin^2(\alpha)\overline{W}_{k_2, r_2} W_{j_2, s_2}a^*_{r_2}a_{s_2}\\ \nonumber
&-\overline{W}_{k_1, 0}W_{j_2, 0}\sum_{r_2\neq 0, s_1\neq 0 }\sigma\sin^2(\alpha)\overline{W}_{k_2, r_2} W_{j_1, s_1}a^*_{r_2}a_{s_1}\\ \nonumber
&-\overline{W}_{k_2, 0}W_{j_1, 0}\sum_{r_1\neq 0, s_2\neq 0 }\sigma\sin^2(\alpha)\overline{W}_{k_1, r_1} W_{j_2, s_2}a^*_{r_1}a_{s_2}. \nonumber
\end{align}
By induction based on the arguments above, Equation (\ref{intermr2}) holds with $\tau^1 (a_{k_1}^*a_{k_2}^*a_{j_2}a_{j_1})$ replaced by $\tau^{t+1} (a_{k_1}^*a_{k_2}^*a_{j_2}a_{j_1})$ and $a_{r_1}^*a_{r_2}^*a_{s_2}a_{s_1}$, respectively $a_{r}^*a_{s}$, replaced by
$\tau^t (a_{r_1}^*a_{r_2}^*a_{s_2}a_{s_1})$, respectively $\tau^t (a_{r}^*a_{s})$.

Note that $\sigma\sin^2(\alpha)\overline{W}_{k, 0}W_{j, 0}=(  W\Ee  W^*)_{j,k}$, 
and with $P^\perp={P^\perp}^*=\overline{P^\perp}$,
\begin{align}
\sum_{r\neq 0, s\neq 0 }\overline{W}_{k, r} W_{j, s}a^*_{r}a_{s}&=
\sum_{r, s}(  W P^\perp)^*_{r, k} a^*_{r}a_{s} ( WP^\perp)_{j, s}.
\end{align}
Hence, with the definitions (\ref{denistyreduced}), (\ref{wedge}) and (\ref{p2abp2}), the LHS in (\ref{intermr2}) equals 
$\bra j_1\wedge j_2 | \rho^{(2)}_1  \, k_1\wedge k_2 \ket$, the  first sum of the RHS equals
\begin{align}
\bra j_1\wedge j_2 |  ( W \K \otimes  W \K) \, \rho_0^{(2)}\, ( W \K \otimes  W \K)^*  \, k_1\wedge k_2 \ket,
\end{align}
and sum of the last four terms equal
\be
2 \bra j_1\wedge j_2 |   ( W\Ee  W^*)\otimes  ( WP^\perp)\, \rho_0^{(1)}\, ( WP^\perp)^*  \, k_1\wedge k_2 \ket .
\ee
This yields the result for $t=0$, and the remark above yields the result for all $t\in\N$.\ep
\end{proof}

In a similar fashion as for the one-body density matrix, the long time limit of the two-body density matrix is independent of the details of the dynamics in the sample and of the coupling, as long as a certain spectral hypothesis ensuring enough mixing is satisfied. 
\begin{cor} \label{47} Assume $\sigma(\M )\cap \Ss=\emptyset$. Then 
\be
\lim_{t\ra\infty}\rho^{(2)}_t=\P^{(2)}\sigma \idtyty\otimes \sigma \idtyty \P^{(2)}=\sigma^2\idtyty_{\cH_s^{\wedge 2}}.
\ee
\end{cor}
\begin{rem}
This result coincides with (\ref{r1infinity}) for $p=2$. 
\end{rem}
\begin{proof}
We first show that $\sigma^2\idtyty_{\cH^{\wedge 2}}$ is a solution of the equation the limiting two-body matrix must satisfy if it exists. Then show that this solution is unique and finally that the limit exists under our hypotheses.
Equation (\ref{tdtbdm}) for $t=\infty$ with $\rho^{(1)}_\infty=\sigma \idtyty_{\cH_s}$ yields the relation
\begin{align}\label{eqrh2inf}
\rho_\infty^{(2)}=\cP_A^{(2)} \M ^{\otimes 2}\cP_A^{(2)} \rho_\infty^{(2)} \cP_A^{(2)} {\M^*}^{\otimes 2}\cP_A^{(2)}
+ 2\sigma \cP_A^{(2)} \B \otimes (\N \N^* )\cP_A^{(2)},
\end{align}
to be satisfied by $\rho_\infty^{(2)} $, assuming its existence. Inserting the Ansatz $\sigma^2\idtyty_{\cH_s^{\wedge 2}}$,  we need to see that the following holds true
\begin{align}\label{idinf2}
\P^{(2)}\idtyty\otimes \idtyty \P^{(2)}= \cP_A^{(2)} \M^{\otimes 2}\cP_A^{(2)} \idtyty\otimes \idtyty  \cP_A^{(2)} {\M^*}^{\otimes 2}\cP_A^{(2)}+\frac{2}{\sigma}\cP_A^{(2)} \B\otimes (\N\N^*)\cP_A^{(2)}.
\end{align}
We make use of the definitions of $\M=W\cK$, $\B/\sigma=\sin^2(\alpha)W|e_0\ket\bra e_0|W^*$, $\N=WP^\perp$,  and of the property 
\be
\P^{(2)}A\otimes B \P^{(2)} C\otimes D \P^{(2)}=\P^{(2)}\frac12(A\otimes B+B\otimes A)\frac12(C\otimes D+D\otimes C)\P^{(2)},
\ee
the RHS of (\ref{idinf2}) reads
\begin{align}\label{rhsid2}
\cP_A^{(2)} W\otimes W\{\cK^2\otimes \cK^2+\sin^2(\alpha)(|e_0\ket\bra e_0|\otimes P^\perp + P^\perp\otimes |e_0\ket\bra e_0|)\}W^*\otimes W^* \cP_A^{(2)}.
\end{align}
Now, $\cK^2=\cos^2(\alpha)|e_0\ket\bra e_0| +P^\perp$, so that the bracket above equals
\begin{align}
&\{ \cos^4(\alpha)|e_0\ket\bra e_0|\otimes |e_0\ket\bra e_0| + P^\perp\otimes P^\perp+(|e_0\ket\bra e_0|\otimes P^\perp + P^\perp\otimes |e_0\ket\bra e_0|)\}=\nonumber\\
&\{\idtyty\otimes \idtyty +(\cos^4(\alpha)-1) |e_0\ket\bra e_0|\otimes |e_0\ket\bra e_0| \}.
\end{align}
As $W$ is unitary and thanks to the identity which holds for any $\ffi, \psi \in\cH_s$
\be
\cP_A^{(2)} |\ffi\ket\bra\psi |\otimes  |\ffi\ket\bra\psi | \cP_A^{(2)} =\mathbb{O},
\ee
we get that (\ref{rhsid2}) equals
\be
\cP_A^{(2)} \{\idtyty\otimes \idtyty +(\cos^4(\alpha)-1) |We_0\ket\bra We_0|\otimes |We_0\ket\bra We_0| \}\cP_A^{(2)}=\cP_A^{(2)}\idtyty\otimes\idtyty \cP_A^{(2)},
\ee
which is what we were aiming for. 

Now, by our spectral hypothesis on the contraction $\M$, the operator  $ \cM^{\wedge 2}$ on $\cH^{\wedge 2}$  
defined by $ A\mapsto \cP_A^{(2)} \M^{\otimes 2}\cP_A^{(2)} A \cP_A^{(2)} {\M^*}^{\otimes 2}\cP_A^{(2)}$ is such that 
$\cM^{\wedge 2}-\idtyty_{\cH^{\wedge 2}}$ is invertible. This is enough to get that the solution to equation (\ref{eqrh2inf}) is unique.

Finally, we prove the existence of the long time limit for $\rho^{(2)}_t$ given by formula (\ref{rh2att}). By our spectral assumption on the contraction $\M$, $(\cP_A^{(2)} \M^{\otimes 2}\cP_A^{(2)})^t\ra {\mathbb O}$ exponentially fast in $t$. Moreover,  $\rho_t^{(1)}$ is uniformly bounded in $t\in\N$, so that a Cauchy sequence argument yields the existence of 
$\lim_{t\ra\infty}\rho^{(2)}_t= \rho^{(2)}_\infty= \sigma^2\idtyty_{\cH^{\wedge 2}}$. \ep
\end{proof}

\subsection{Application to Quantum Walks}

We shall consider here perturbed unitary dynamics on the sample given by one of the simplest instances of quantum walks on the discrete circle, the so-called coined quantum walks.  While we restrict attention to this case, it is very likely that our results also apply to other one dimensional dynamics considered in \cite{BHJ}, \cite{HJS}, \cite{ABJ2}, under suitable hypotheses.

Let us briefly recall the setup: the Hilbert space of the quantum walker is $\C^2\otimes l^2(\{0,1,\dots, n-~1\})$, where $\C^2$ is the spin or coin space and $\{0,1,\dots, n-1\}$ is the configuration space of the quantum walker. The canonical basis in 
$\C^2\otimes l^2(\{0,1,\dots, n-1\})$ is denoted by $\{|\tau\otimes x\ket\}_{\tau=\pm 1, x\in \{0,1,\dots, n-1\}}$ and we write the orthogonal projectors on the spin states by  $P_\tau=|\tau\ket\bra\tau |$, $\tau\in\pm 1$. Given a configuration of unitary matrices $\cC=\{C_x\}_{x\in\{0,1,\dots, n-1\}}$ on $\C^2$, the spin or coin matrices, the one-time step unitary dynamics $V(\cC)$ on $\C^2\otimes l^2(\{0,1,\dots, n-1\})$  is defined by 
\begin{equation}\label{defqw}
V(\cC) = \sum_{x\in \{0,1,\dots, n-1\}} \left\{ P_{+1} C_x\otimes|x+1\rangle\langle x| + P_{-1} C_x\otimes|x-1\rangle\langle x|\right\}, 
\end{equation}
with periodic boundary conditions in the configuration space: $|n\ket\equiv|0\ket$, $|-1\ket\equiv |n-1\ket$. This corresponds to first acting on the spin part of the quantum walker sitting at site $x$ with the spin matrix $C_x$, and then shifting by one step to the right or left, depending on the spin component.We speak of random quantum walks in case the spin matrices $\{C_k\}_{k\in\{0,1,\dots, n-1\}}$ are unitary matrix valued random variables, which gives rise to Anderson localization phenomena, see \cite{JM} and \cite{ASWe} for one dimensional results and \cite{J3} for arbitrary dimensions.

To fit in the framework we used so far, we implement the unitary isomorphism $\C^2\otimes l^2(\{0,1,\dots, n-~1\})\simeq l^2(\{0,1,\dots, d-~1\})$, with $d=2n$, using the following map of ordered canonical bases
\be
\{|+1\otimes 0\rangle, |-1\otimes 0\rangle, |+1\otimes 1\rangle, |-1\otimes 1\rangle, \ldots, |+1\otimes n-1\rangle, |-1\otimes n-1\rangle\}= \{e_0, e_1, \dots, e_{d-1}\}.
\ee
With $C_x=\begin{pmatrix} \alpha_x & \beta_x \cr \gamma_x & \delta_x \end{pmatrix}$,  the matrix representation $W$ of $V(\cC)$ in that basis reads
\begin{align}
W=\begin{pmatrix}
0 & 0 &  &  &  &  \alpha_{n-1}& \beta_{n-1}  \cr 
0 & 0 & \gamma_1 & \delta_1 &  &  &      \cr 
\alpha_0 & \beta_0 & 0 & 0 & & & \cr
 &  &0  & 0 &  &  &      \cr
  &  &  \alpha_1 & \beta_1  & \ddots &  \gamma_{n-1} & \delta_{n-1}       \cr 
 & &  &  &  &  0 & 0       \cr 
 \gamma_0 & \delta_0 &  &  &  &  0 & 0       
\end{pmatrix}.
\end{align}
The unitary matrix $W$ yields the one-body one time step evolution in the sample that appears in the definitions
\be
\M=W\K, \ \ \B=W\Ee W^*, \ \ \N=WP^\perp,
\ee
where $\cK=\idtyty-(1-\cos(\alpha))|e_0\ket\bra e_0|$, $P^\perp=\idtyty - |e_0\ket\bra e_0|$. 
Using the notation $c\equiv \cos(\alpha)$, the contraction $\M$ has the form 
\be
\M=\begin{pmatrix}
0 & 0 &  &  &  &  \alpha_{n-1}& \beta_{n-1}  \cr 
0 & 0 & \gamma_1 & \delta_1 &  &  &      \cr 
c\alpha_0 & \beta_0 & 0 & 0 & & & \cr
 &  &0  & 0 &  &  &      \cr
  &  &  \alpha_1 & \beta_1  & \ddots &  \gamma_{n-1} & \delta_{n-1}       \cr 
 & &  &  &  &  0 & 0       \cr 
 c\gamma_0 & \delta_0 &  &  &  &  0 & 0       
\end{pmatrix}.
\ee
\begin{lem}\label{lemcyc} Let $\cos(\alpha)\notin\{1, -1\}$ and assume $e_0$ is cyclic for $W$. Then $\sigma(\M)\cap \Ss=\emptyset$.
\end{lem}
\begin{rems} i) The condition on $\cos(\alpha)$ ensures $\M$ is not unitary, whereas the cyclicity of $e_0$ ensures there is no invariant subspace on which $\K$ acts like the identity.  \\
ii) The latter implies the spectrum of $W$ is simple with eigenvectors all having non zero component along $e_0$.  \\
iii) These hypotheses holds if $0<|\alpha_x\beta_x|$, for all $x\in \{0,\dots, n-1\}$. Such quantum walks are generic.
\end{rems}
\begin{proof}
Setting $\lambda=1-\cos(\alpha)\notin \{0,2\}$, we can write 
$
\M =W (\idtyty - \lambda |e_0\ket\bra e_0|).
$ 
Hence $\M$ is a rank one perturbation of $W$, whose resolvent reads
\be\label{resom}
(\M-z)^{-1}=\left(\idtyty +\frac{\lambda(W-z)^{-1}W|e_0\ket\bra e_0|}{1-\lambda\bra e_0|((W-z)^{-1}We_0\ket)}\right)(W-z)^{-1}
\ee
for all $z$ in the resolvent set of $W$, $\rho(W)$, and such that $1-\lambda\bra e_0|((W-z)^{-1}We_0\ket)\neq 0$. We need to show that $(\M-z)^{-1}$ is regular on the unit circle. 

Consider first the possibility that $z=e^{i\theta}\in \rho(W)$ is a zero of the denominator. For $\sigma(W)=\{e^{i\alpha_j}\}_{j=0,\dots, d-1}$, with normalized eigenvectors $\{v_{j}\}_{j=0,\dots, d-1}$, this is equivalent to
\be
1/\lambda=\sum_{j=0}^{d-1}|\bra v_j|e_0\ket|^2\frac{1}{1-e^{i(\theta -\alpha_j)}}\in\R^+.
\ee
But the real part of the RHS equals $1/2$, which corresponds to $\lambda=2$, a forbidden value.

Then, take $z$ in a neighbourhood of $e^{i\alpha_{j_0}}$, some eigenvalue of $W$. By assumption, 
\be
(W-z)^{-1}=\frac{P_{\alpha_j}}{(e^{i\alpha_{j_0}}-z)}+\cO({1}), 
\ee
with $P_{\alpha_j}=|v_{j}\ket\bra v_j |$. Similarly, as $\bra e_0|P_{\alpha_{j_0}}e_0\ket \neq 0$,
\be
\frac{\lambda(W-z)^{-1}W|e_0\ket\bra e_0|}{1-\lambda\bra e_0|((W-z)^{-1}We_0\ket)}=\frac{\lambda (P_{\alpha_{j_0}}e^{i\alpha_{j_0}}+\cO({e^{i\alpha_{j_0}}-z}))|e_0\ket\bra e_0|}{-\lambda \bra e_0|P_{\alpha_{j_0}}e_0\ket e^{i\alpha_{j_0}}+\cO({e^{i\alpha_{j_0}}-z}) },
\ee
so that the leading term as $z\ra e^{i\alpha_{j_0}}$ of (\ref{resom}) is
\be
\left(\idtyty+\frac{\lambda P_{\alpha_{j_0}}e^{i\alpha_{j_0}}|e_0\ket\bra e_0|}{-\lambda \bra e_0|P_{\alpha_{j_0}}e_0\ket e^{i\alpha_{j_0}} }\right)\frac{P_{\alpha_j}}{(e^{i\alpha_{j_0}}-z)}= \left(\frac{\bra e_0|P_{\alpha_{j_0}}e_0\ket -P_{\alpha_{j_0}}|e_0\ket\bra e_0|}{ \bra e_0|P_{\alpha_{j_0}}e_0\ket }\right)\frac{P_{\alpha_j}}{(e^{i\alpha_{j_0}}-z)},
\ee
where
\be
(\bra e_0|P_{\alpha_{j_0}}e_0\ket -P_{\alpha_{j_0}}|e_0\ket\bra e_0|)P_{\alpha_{j_0}}=   |\bra v_{j_0}|e_0\ket |^2   |v_{j_0}\ket\bra v_{j_0}|
- |v_{j_0}\ket  |\bra v_{j_0}|e_0\ket |^2 \bra v_{j_0}|=0.
\ee
Hence, $(\M-z)^{-1}$ is analytic in a neighbourhood of $\Ss$, which implies that $\sigma (\M)\subset \D$, $\D$ the open unit disk.
\ep
\end{proof}

In order to make it clear that the periodic boundary conditions used in the definition of the quantum walk play no role, we note that if we consider the quantum walk (\ref{defqw}) with fixed spin matrices $C_0=C_{n-1}=\begin{pmatrix}0& 1\cr 1& 0\end{pmatrix}$ at the boundary, one sees that the two dimensional subspace $\mbox{span}\{+1\otimes 0\rangle, |-1\otimes n-1\rangle\}$ and its orthogonal complement are both invariant under $V(\cC)$. Thus, when restricted to $\mbox{span}\{ |-1\otimes 0\rangle, |+1\otimes 1\rangle, |-1\otimes 1\rangle, \ldots,,|+1\otimes n-2\rangle, |-1\otimes n-2\rangle, |+1\otimes n-1\rangle\}$, $V(\cC)$ defines a quantum walk with boundary conditions at $0$ and $n-1$ which makes the quantum walker bounce back when it meets the boundary $\{0\}\cup\{n-1\}$ of the configuration space. Therefore we can consider the latter quantum walk is supplemented by Dirichlet boundary conditions.
Using now the mapping of canonical basis (with $d=2(n-1)$),
\be
\{|-1\otimes 0\rangle, |+1\otimes 1\rangle, |-1\otimes 1\rangle, \ldots, |-1\otimes n-2\rangle, |+1\otimes n-1\rangle\}= \{e_0, e_1, \dots, e_{d-1}\},
\ee
the matrix representation $W$ of the unitary operator $V(\cC)|_{\mbox{span}\{+1\otimes 0\rangle, |-1\otimes n-1\rangle\}^\perp}$ in the latter bases reads
\begin{align}
W=\begin{pmatrix}
 0 & \gamma_1 & \delta_1 &  &    & &   \cr 
1 & 0 & 0  & & & &  \cr
    &0  & 0 & &\gamma_{n-2} & \delta_{n-2}&    \cr
  &  \alpha_1 & \beta_1  & \ddots &0 & 0&   \cr 
   &   &  & & 0 & 0&1          \cr 
     &  &  &  & \alpha_{n-1} & \beta_{n-2}&0       \cr
\end{pmatrix}.
\end{align}
Hence Lemma \ref{lemcyc} applies to the matrix $\M=W\K$ corresponding this matrix $W$ as well, and $e_0$ is cyclic for $W$ under the generic condition $0<|\alpha_x\beta_x|$, for all $x\in \{1,\dots, n-2\}$.

Therefore, as a direct corollary, we get thermalization of generic quantum walks with periodic and Dirichlet boundary conditions:
\begin{cor}\label{corrqw}
The one-body and two-body reduced density matrices for fermionic generic quantum walks tends to $\sigma\idtyty_{\cH}$ and $\sigma^2\idtyty_{\cH^{\wedge 2}}$  as time goes to infinity, independently of the details of the walk. In particular, as time goes to infinity, the density profile in the sample is flat: $\bra n_j\ket=\sigma$ and the correlations are constant: $\bra n_j^sn_k^s\ket=\sigma^2$.
\end{cor}

The results above hold for any size of the sample given by $d$, and for any generic quantum walk, be it random or deterministic. In particular, the random quantum walks considered in \cite{JM} satisfy the hypotheses of the Corollary and are known to display dynamical localization. Hence, if we assume there is no particle in the sample initially, the time evolution populates the sample through the site zero. Corollary \ref{corrqw} shows that dynamical localization cannot prevent the particles from visiting the whole sample when the reservoir keeps interacting with the sample, and moreover, the thermalization process at work washes out any spatial structure in the asymptotic repartition of particles in the sample.

\end{document}